\documentclass[journal]{IEEEtran} 

\usepackage{bbm} 

\usepackage{bm} 

\usepackage{amsmath} 

\usepackage{amsthm} 

\usepackage{amssymb} 

\usepackage{mathrsfs} 

\usepackage{optidef}  

\usepackage{algorithm} 

\usepackage[noend]{algpseudocode} 

\usepackage{cite} 

\usepackage{graphicx} 

\usepackage{textcomp} 

\usepackage{xcolor} 

\usepackage{caption} 
\usepackage{subcaption} 

\usepackage{hyperref} 

\usepackage{tabularx,booktabs}
\newcolumntype{C}{>{\centering\arraybackslash}X} 
\usepackage{tikz}
\newcommand{\Cross}{$\mathbin{\tikz [x=1.5ex,y=1.5ex,line width=0.5ex, red] \draw (0,0) -- (1,1) (0,1) -- (1,0);}$}
\newcommand{\Checkmark}{$\color{green}\checkmark$}
\newcommand{\Circle}{$\mathbin{\tikz [x=1.1ex,y=1.1ex,line width=2ex, orange] \draw (0,0) circle (1pt);}$}

\newtheorem{theorem}{Theorem}
\newtheorem{lemma}{Lemma}

\newtheorem{corollary}{Corollary}
\newtheorem{definition}{Definition}

\newtheorem{remark}{Remark}

\newtheorem{condition}{Condition}

\newcommand\ddfrac[2]{\frac{\displaystyle #1}{\displaystyle #2}}

\DeclareMathOperator*{\argmin}{argmin}

\begin{document}
\title{Preempting to Minimize Age of Incorrect Information under Transmission Delay}

\author{
\IEEEauthorblockN{Yutao Chen and Anthony Ephremides}\\
\IEEEauthorblockA{Department of Electrical and Computer Engineering, University of Maryland}}

\maketitle

\begin{abstract}
We study the problem of optimizing the decisions of a preemptively capable transmitter to minimize the Age of Incorrect Information (AoII) when the communication channel has a random delay. We consider a slotted-time system where a transmitter observes a Markovian source and makes decisions based on the system status. In each time slot, the transmitter decides whether to preempt or skip when the channel is busy. When the channel is idle, the transmitter decides whether to send a new update. A remote receiver estimates the state of the Markovian source based on the update it receives. We consider a generic transmission delay and assume that the transmission delay is independent and identically distributed for each update. This paper aims to optimize the transmitter's decision in each time slot to minimize the AoII with generic time penalty functions. To this end, we first use the Markov decision process to formulate the optimization problem and derive the analytical expressions of the expected AoIIs achieved by two canonical preemptive policies. Then, we prove the existence of the optimal policy and provide a feasible value iteration algorithm to approximate the optimal policy. However, the value iteration algorithm will be computationally expensive if we want considerable confidence in the approximation. Therefore, we analyze the system characteristics under two canonical delay distributions and theoretically obtain the corresponding optimal policies using the policy improvement theorem. Finally, numerical results are presented to illustrate the performance improvements brought about by the preemption capability.
\end{abstract}

\begin{IEEEkeywords}
Age of Incorrect Information (AoII), information freshness, semantic communications, delay, preemption
\end{IEEEkeywords}

\section{Introduction}
As communications technologies evolve, so do the demands on communications networks. For example, we need communications networks to be more efficient and intelligent. At the same time, we question whether traditional performance metrics such as delay can still meet these higher demands. Therefore, researchers have recently proposed Semantic Communications~\cite{b1}, a new design paradigm for networked systems. Semantic communications consider the semantics of the information, defined as the importance of the transmitted information for the purpose of transmission. Semantic measures are at the core of semantic communications. In~\cite{b1}, the authors present several representative semantic measures. The first is freshness, which captures how fresh the information is. In other words, it measures the time elapsed between the generation of the latest information at the destination and its arrival at the destination. Age of Information (AoI), introduced in~\cite{b2}, is a good and widely studied example~\cite{b22,b23,b24}. The second is relevance, which captures the change in the process since the last sampling and is very important in remote estimation. For example, when the process changes slowly, we can reduce the transmission of information to save valuable resources. When the process changes dramatically, we need to increase the transmission of information so that the distant receiver can have better knowledge of the process. Relevance is different from freshness since freshness ignores the specific content of the transmitted information. Hence, as shown in~\cite{b3}, optimizing AoI does not achieve optimal system performance for communication purposes. The third is value, which captures the value of information transmission for communication purposes. Value of Information (VoI) is a good example that quantifies the difference between the benefit of transmitting this information and its cost. One of the essential parts of semantic communications is the design of semantic metrics that quantify semantic measures. We note that AoI and VoI capture only one part of the semantic measure. However, a more refined semantic metric needs to consider more than one semantic measure and integrate multiple semantic measures well into a single semantic metric. The Age of Incorrect Information (AoII) introduced in~\cite{b4} is an example.

AoII combines freshness and relevance of information. As presented in~\cite{b4}, AoII is dominated by two penalty functions. The first one is the time penalty function, which is based on the idea that information only ages if it is incorrect. If information provides accurate information, we consider it fresh no matter when it was generated. Through the time penalty function, AoII captures the time elapsed since the last time the receiver had the correct information. The second is the information penalty function, which is based on the idea that different information mismatches will cause different damage to the system. When the difference between the information on the receiver and the correct information is slight, it is unlikely to cause significant damage to the system. On the other hand, a significant difference will cause large damage to the system. Therefore, the information penalty function captures the mismatch between the information on the receiver side and the correct information. By combining the two penalty functions, AoII not only captures the mismatch between the information on the receiver side and the correct information but also reflects the aging process of the incorrect information.

With the introduction of AoII, researchers have devoted themselves to revealing its characteristics and performance in networked systems. In~\cite{b4}, the authors study the minimization of AoII in the presence of average transmission rate limits. Then, in~\cite{b5}, the result is extended to the case when the time penalty function is general. In~\cite{b6}, the authors study a system setup similar to the one described above but with a source process with multiple states and an AoII that incorporates the quantified information mismatches between the source and receiver. Unlike the previous papers, \cite{b7} studies AoII in the context of scheduling. In this type of problem, a base station sends updates to multiple users and tries to ensure that each user's information is as accurate as possible. In~\cite{b7}, the authors study the problem of minimizing AoII when channel state information is available and the time penalty function is generic. The authors of~\cite{b8} consider a similar system, but the base station has no way of knowing the actual state of the event before deciding to transmit. In the real world, we usually do not know the statistical model of the dynamic process we need to observe, or we do not know the model's parameters. Minimizing AoII, in this case, is usually tricky. Therefore, the authors of~\cite{b20} consider the problem of minimizing AoII without knowing the parameters of a Markovian process. Moreover, a variant of AoII - Age of Incorrect Estimates is introduced and studied in~\cite{b21}. In most papers studying AoII, the update transmission time is constant, usually one time slot. In practice, however, the communication channel usually suffers from a random delay due to the influence of various factors. In such a system setting, the authors of~\cite{b9} compare three performance metrics: AoII, AoI, and real-time error through extensive numerical results. \cite{b10} and~\cite{b11} study the problem of minimizing AoII under random channel delay from a theoretical perspective. However, in both papers, the transmitter must wait for the update to complete the transmission before initiating the subsequent transmission. In this paper, we study the case where the transmitter can preempt the transmitting updates and immediately transmit a new update. In this way, we also have to consider whether the transmitter needs to terminate the transmission to transmit new updates when the channel is busy. 

The main contributions of this paper can be summarized as follows. 1) We study the problem of optimizing the AoII with a generic time penalty function in a slotted-time system under a generic transmission delay. 2) We formulate the problem using the Markov decision process and prove the existence of the optimal policy. 3) We obtain the analytical expressions of the expected AoIIs achieved by two canonical preemptive policies. 4) We propose the value iteration algorithm that approximates the optimal policy. 5) We analyze two canonical transmission delays and theoretically find the corresponding optimal policies. Specifically, we first study the case where the transmission delay follows the Geometric distribution, a typical example of an unbounded transmission time. Then, we investigate the case where the transmission delay follows the Zipf distribution, a typical example of a bounded transmission time. We also extend the results to the generic transmission delay when the transmission time is upper bounded by $2$.

The paper is organized as follows. Section~\ref{sec-SystemOverview} describes the system model, specifies the choice of the penalty functions in AoII, and formulates the optimization problem. Then, we cast the optimization problem into a Markov decision process and specify the state transition probabilities in Section~\ref{sec-MDP}. In Section~\ref{sec-Performance}, we analyze two preemptive policies and obtain the analytical expressions of the expected AoIIs they achieve. Then, in Section~\ref{sec-OptimalPolicy}, we prove the existence of the optimal policy, propose a modified value iteration algorithm to approximate the optimal policy, and theoretically obtain the optimal policy when the transmission delay follows two canonical distributions. The paper concludes with numerical results detailed in Section~\ref{sec-Numerical}.

\section{System Overview}\label{sec-SystemOverview}
\subsection{System Model}\label{sec-SystemModel}
We consider a system in which the transmitter observes a Markovian source by receiving updates from it and controls the transmission of the updates over a communication channel with random delays so that the receiver at the other end of the channel has the best real-time knowledge of the Markovian source. We assume that time is slotted and normalized to a unit time slot. Meanwhile, the end of one time slot is the beginning of the next time slot. An illustration of the system model is given in Fig.~\ref{fig-SystemModel}.
\begin{figure}[!t]
\centering
\includegraphics[width=3in]{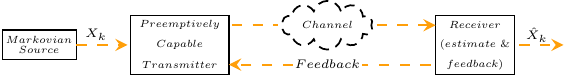}
\caption{An illustration of the system model.}
\label{fig-SystemModel}
\end{figure}
At the beginning of time slot $k$, the transmitter receives an update $X_k$ from the Markovian source and discards the old one. To not overcomplicate the system model, we assume that the Markovian source has two states and is symmetric with state transition probability $p$. Then, the transmitter decides whether to transmit $X_k$ based on the system's current status. When the channel is idle, the transmitter can choose whether to transmit the current update. We consider the case where the transmitter has the ability to terminate an ongoing transmission and immediately initiate a new one. We assume that terminating the current transmission and starting a new one can be done simultaneously, and the preempted update will be discarded. Thus, when the channel is busy, the transmitter can stay idle or terminate the ongoing transmission and immediately transmit the current update. We denote the action of the transmitter by $a_k$. $a_k=0$ when the transmitter chooses to stay idle. Otherwise, $a_k=1$. Note that the specific action represented by $a_k=1$ depends on the status of the communication channel. If not preempted, the transmission will arrive at the receiver after a random amount of time slots. The receiver maintains an estimate of the state of the Markovian source based on the received updates. Whenever the receiver receives a new update, it sends a feedback signal to the transmitter to inform the transmitter that the update has been received.

In the following, we define the delay model. The communication channel is reliable but suffers from random delays, which means that if an update is not preempted, it will be delivered to the receiver losslessly after a random amount of time slots. The transmission time is a random variable denoted by $T$. For simplicity, we assume that $T$ is independent and identically distributed for each update. The transmission time can be fully characterized by the probability mass function (PMF) denoted by $p_t\triangleq Pr(T=t)$, where $t\in\mathbbm{N}^*$. We do not impose any restrictions on the transmission time, which means that the transmission time can be infinite.

Next, we describe the receiver's estimation strategy and the feedback mechanism. Let $\hat{X}_k$ denote the receiver's estimate at time slot $k$. Then, according to~\cite{b12}, the best estimator when $p\leq\frac{1}{2}$ is the last received update. Let $d_k=1$ when an update is delivered to the receiver at time slot $k$ and $d_k=0$ otherwise. Then,
\[
{\hat{X}_k} = 
\begin{dcases}
\hat{X}_{k-1} & d_{k-1}=0,\\
X_{k-T} & d_{k-1}=1,
\end{dcases}
\]
where $X_{k-T}$ is the update delivered at the end of time slot $k-1$. In the case of $p>\frac{1}{2}$, the best estimator depends on the realization of the transmission time. In this paper, we consider only the case of $0<p<\frac{1}{2}$. We exclude the case of $p=0$ because, in this case, the Markovian source never changes state. We also exclude the case of $p=\frac{1}{2}$ because, in this case, the state of the Markovian source is independent of the previous state. The results can be extended to the case of $p>\frac{1}{2}$ by using the corresponding best estimator. Whenever the receiver receives a new update, it sends an $ACK$ packet to the transmitter so that the transmitter is aware of the change in the receiver's estimate and that the channel has become idle. In real-world applications, the size of $ACK$ packets is usually negligible compared to that of status updates. Therefore, we assume that $ACK$ packets are accurately and instantaneously received by the transmitter. This assumption is widely used in the relevant literature~\cite{b4,b13}. The $ACK$ packet alone is sufficient for the transmitter to keep track of the receiver's estimate since we assume that the update will necessarily be delivered unless it is preempted.

\subsection{Age of Incorrect Information}\label{sec-AoII}
The system uses the Age of Incorrect Information (AoII) to measure its performance. AoII is a combination of the information mismatch and the aging process of the mismatched information. Specifically, the information mismatch is quantified by the information penalty function $g(X_k,\hat{X}_k)$, and the aging process of mismatched information is characterized by the time penalty function $f(k)$. Then, in a slotted time system, AoII at time slot $k$ can be written as
\[
\Delta_{AoII}(X_k,\hat{X}_k,k) = \sum_{h = U_k+1}^{k}\bigg(g(X_h,\hat{X}_h)F(h-U_k)\bigg),
\]
where $F(k) \triangleq f(k) - f(k-1)$ and $U_k$ is the last time slot where the receiver's estimate is correct. Mathematically,
\[
U_k \triangleq \max\{h:h\leq k, X_h = \hat{X}_h\}.
\]
To avoid unnecessary complications, we choose $g(X_k,\hat{X}_k) = |X_k-\hat{X}_k|$. Since the Markovian source has two states, $g(X_k,\hat{X}_k)\in\{0,1\}$. Consequently, AoII can be simplified as
\begin{equation}\label{eq-AoIIFirst}
\Delta_{AoII}(X_k,\hat{X}_k,k) = f(k-U_k)\triangleq f(\Delta_k).
\end{equation}
To facilitate the analysis, we make the following assumptions on the time penalty function $f(\Delta)$.
\begin{itemize}
\item $f(\Delta_1)\geq f(\Delta_2)\geq0$ if $\Delta_1\geq\Delta_2$.
\item $f(\Delta)\rightarrow\infty$ when $\Delta\rightarrow\infty$.
\item $\sum_{\Delta=0}^{\infty}\gamma^\Delta f(\Delta) < \infty$ for $0<\gamma<1$.
\end{itemize}
Some choices of the time penalty function $f(\Delta)$ are the following.
\begin{itemize}
\item $f(\Delta) = \alpha\Delta+\beta$ where $\alpha\geq0$ and $\beta\geq0$ are finite constants.
\item $f(\Delta)=\kappa\Delta^2$ where $\kappa\geq0$ is finite constant.
\item $f(\Delta) = \log_a(\Delta+1)$ where $a>1$ is finite constant.
\end{itemize}
A sample path when $f(\Delta_k) = 2\Delta_k$ is given in Fig.~\ref{fig-SamplePath}.
\begin{figure}[!t]
\centering
\includegraphics[width=3in]{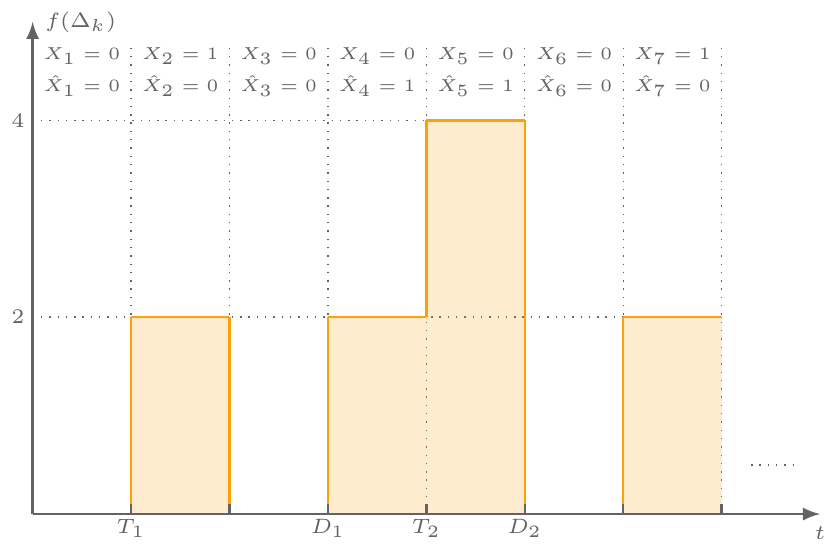}
\caption{A sample path when $f(\Delta_k) = 2\Delta_k$. In the figure, $T_i$ and $D_i$ are the transmissions start and finish time of the $i$-th update, respectively. For example, at $T_1$, the transmission of $X_2$ starts, and the update is delivered at $D_1$. Therefore, the estimate at time slot $4$ (i.e., $\hat{X}_4$) is changed. Note that the transmission decisions in the figure are random.}
\label{fig-SamplePath}
\end{figure}
We notice that the evolution of $\Delta_k$ can fully characterize the evolution of $f(\Delta_k)$. Leveraging the definition of $U_k$, the evolution of $\Delta_k$ can be characterized by the following two cases.
\begin{itemize}
\item When the receiver's estimate is correct at time slot $k$, $U_k = k$ by definition. Hence, $\Delta_k = 0$.
\item When the receiver's estimate is erroneous at time slot $k$, $U_k = U_{k-1}$ by definition. Hence, $\Delta_k = k-U_k = \Delta_{k-1} + 1$.
\end{itemize}
The evolution of $\Delta_k$ can be summarized as follows.
\[
\Delta_{k} = \mathbbm{1}\{X_k\neq \hat{X}_k\}(\Delta_{k-1}+1),
\]
where $\mathbbm{1}\{A\}$ is an indicator function that takes the value of $1$ when event $A$ occurs and $0$ otherwise. We can conclude that $\Delta_k=0$ if and only if $X_{k} = \hat{X}_k$. For later use in analysis, we reformulate the evolution of $\Delta_k$ by incorporating the dynamics of the Markovian source. To this end, we first define $\Gamma\triangleq\mathbbm{1}\{(\Delta_{k-1} = 0\land \hat{X}_{k} = \hat{X}_{k-1})\lor(\Delta_{k-1} > 0\land\hat{X}_{k} \neq \hat{X}_{k-1})\}$.\footnote{In the definition, $\land$ is the logical AND operator and $\lor$ is the logical OR operator.} Then, we know that $\hat{X}_{k} = X_{k-1}$ if $\Gamma=1$ and $\hat{X}_{k} \neq X_{k-1}$ if $\Gamma=0$. Hence, we have
\begin{equation}\label{eq-AoIIDynamic}
{\Delta_{k}} =
\begin{dcases}
\mathbbm{1}\{X_k\neq X_{k-1}\}(\Delta_{k-1}+1) & \Gamma=1,\\
\mathbbm{1}\{X_k= X_{k-1}\}(\Delta_{k-1}+1) & \Gamma=0.
\end{dcases}
\end{equation}
Note that the system's state at any time slot can correspond to either $\Gamma=1$ or $\Gamma=0$.

\subsection{Problem Formulation}
In this paper, we investigate the problem of minimizing the AoII by controlling the transmitter's decision in each time slot. We define a policy as one that specifies the transmitter's decision in each time slot based on the current system status. Then, this paper aims to find a policy that minimizes the AoII of the system. Mathematically, the problem can be formulated as the following minimization problem.
\begin{argmini}|l|
{\psi \in \Psi} {\lim_{K\to\infty} \frac{1}{K}\mathbb{E}_{\psi}\left(\sum_{k=0}^{K-1}f(\Delta_k)\right),}{}{\label{eq-Goal}}
\end{argmini}
where $\mathbb{E}_{\psi}$ is the conditional expectation, given that policy $\psi$ is adopted and $\Psi$ is the set of all admissible policies.
\begin{definition}[Optimal policy]
A policy is optimal if it minimizes the AoII of the system. The action specified by the optimal policy is called the optimal action. We use $\psi^*$ and $a^*$ to denote the optimal policy and action, respectively.
\end{definition}
In the next section, we characterize the optimization problem \eqref{eq-Goal} using the Markov decision process (MDP).

\section{MDP Characterization}\label{sec-MDP}
We use an infinite horizon with average cost MDP $\mathcal{M}$ to characterize the minimization problem \eqref{eq-Goal}. Specifically, $\mathcal{M}$ consists of the following components.
\begin{itemize}
\item The state space $\mathcal{S}$. The state $s=(\Delta,t,i)$ where $\Delta\in\mathbb{N}^0$ is the $\Delta_k$ defined in \eqref{eq-AoIIFirst} without the time stamp. $t\in\mathbb{N}^0$ denotes the time that the transmission has been in progress. When the channel is idle, we define $t=0$. $i\in\{-1,0,1\}$ indicates the channel status. When the channel is idle, $i=-1$. When the channel is busy transmitting an update, $i\in\{0,1\}$, where $i=0$ if the update being transmitted is the same as the receiver's current estimate. Otherwise, $i=1$. To better distinguish between different states, we will use $s$ and $(\Delta,t,i)$ interchangeably to represent the state throughout the rest of the paper. Therefore, $s$ and $(\Delta,t,i)$ will synchronize all changes, such as adding superscripts or subscripts.
\item The action space $\mathcal{A}$. The feasible action is $a\in\{0,1\}$. When $i\neq-1$, $a=1$ if the transmitter decides to terminate the current transmission and immediately start a new one. Otherwise, $a=0$. When $i=-1$, $a=1$ if the transmitter decides to transmit the new update, and $a=0$ otherwise.
\item The state transition probability $\mathcal{P}$. The probability that action $a$ at state $s$ leads to state $s'$ is denoted by $P_{s,s'}(a)$. The value of $P_{s,s'}(a)$ will be discussed in the next subsection.
\item The immediate cost $\mathcal{C}$. The immediate cost for being at state $s$ is $C(s)=f(\Delta)$.
\end{itemize}
Let $V(s)$ be the value function of state $s$. Then, the optimal action at state $s$, denoted by $a^*(s)$, can be determined by the following equation.
\[
a^*(s) = \argmin_{a\in\mathcal{A}}\left\{\sum_{s'\in\mathcal{S}}P_{s,s'}(a)V(s')\right\}\quad s\in\mathcal{S}.
\]
Hence, computing the value function for each state $s\in\mathcal{S}$ is sufficient to obtain the optimal policy. It is well known that $V(s)$ satisfies the Bellman equation.
\[
V(s) + \theta = \min_{a\in\mathcal{A}}\left\lbrace C(s) + \sum_{s'\in\mathcal{S}}P_{s,s'}(a)V(s')\right\rbrace \quad s\in\mathcal{S},
\]
where $\theta$ is the expected AoII achieved by the optimal policy. Hence, the state transition probability $P_{s,s'}(a)$ plays a vital role. In the following, we delve into the expression of $P_{s,s'}(a)$.

We first define and compute an auxiliary quantity $Pr(T=t\mid t-1)$, which is the probability that an update will be delivered in the next time slot, given that the transmission has been in progress for $t-1$ time slots. It is easy to get
\[
Pr(T=t\mid t-1) = \frac{p_{t}}{1-\sum_{i=1}^{t-1}p_i}.
\]
For simplicity, we abbreviate $Pr(T=t\mid t-1)$ as $q_{t}$ for the remainder of this paper. Leveraging $q_t$, we can proceed with deriving the state transition probability $P_{s,s'}(a)$. For the sake of space, the detailed discussion is provided in Appendix~\ref{app-STP} of the supplementary material.

\section{Preemptive Policy Performance}\label{sec-Performance}
In this section, we analyze and evaluate the performance of two preemptive policies by deriving the analytical expressions of the expected AoIIs they achieve. We start with the policy under which the transmitter always preempts the transmitting update and starts a new transmission when the channel is idle.

\subsection{Strong Preemptive Policy}\label{sec-PerformanceSP}
\begin{definition}[Strong preemptive policy]\label{def-SP}
The strong preemptive policy always starts a new transmission when the channel is idle and always preempts the transmitting update.
\end{definition}
\begin{remark}
We consider the strong preemptive policy because it is intuitively desirable when the transmission delay follows a memoryless distribution. For example, the Geometric distribution. One of the essential properties of the Geometric distribution is that $q_t$ is independent of $t$, meaning that no matter how long an update has been in transmission, it has the same probability of being delivered in the next time slot. Therefore, it is desirable for the transmitter to preempt so that the update in the channel is always the freshest.
\end{remark}
The system dynamics under the strong preemptive policy can be fully characterized by a discrete-time Markov chain (DTMC). Without loss of generality, we assume the system starts at state $(0,0,-1)$. Then, the state space of the induced DTMC $\mathcal{S}^{MC}_{sp}$ consists of all the states that are accessible from state $(0,0,-1)$. For a better presentation, we introduce the \textit{virtual state}. By definition, the DTMC will never visit the virtual state. Nevertheless, the existence of these virtual states will make the equations clearer. In the following, we elaborate on the state space $\mathcal{S}^{MC}_{sp}$ and identify the virtual states. We first recall that the strong preemptive policy always preempts the transmitting updates. Hence, each update can only live for one slot in the channel before being preempted or delivered. Consequently, the DTMC will never reach state $s$ with $t>1$. Hence, the system can only be in state $s$ with $t\in\{0,1\}$. With this in mind, $\mathcal{S}^{MC}_{sp}$ consists of the following states.
\begin{itemize}
\item $s=(\Delta,0,-1)$ where $\Delta\geq0$. The DTMC will be in this state every time the channel is idle.
\item $s=(\Delta,1,0)$ where $\Delta\geq0$. The DTMC will be in this state when the channel is busy transmitting an update that is the same as the receiver's estimate. Then, we identify the virtual state. We note that $i=0$ happens only when the transmitter initiates the transmission when AoII is zero. We recall that the transmitting update is either delivered or preempted one time slot after the transmission starts. Combined with the fact that, within one time slot, AoII can either increase by $1$ or decrease to zero, we know that $s=(\Delta,1,0)$ where $\Delta\geq2$ is a virtual state.
\item $s=(\Delta,1,1)$ where $\Delta\geq0$. The DTMC will be in this state when the channel is busy transmitting an update that differs from the receiver's estimate. Then, we identify the virtual state. We notice that $i=1$ occurs when the transmitter initiates the transmission when the AoII is not zero. Combined with the fact that, within one time slot, AoII either increases by $1$ or decreases to zero, we know that $s=(1,1,1)$ is a virtual state.
\end{itemize}
We notice that all the states in $\mathcal{S}^{MC}_{sp}$, except for the virtual states, communicate with state $(0,0,-1)$. Combined with the fact that state $s=(0,0,-1)$ is recurrent, we can conclude that the stationary distribution of the induced DTMC exists. We denote by $\pi_{-1}(\Delta)$ the steady state probability of state $s = (\Delta,0,-1)$. Likewise, we denote the steady state probability of state $s = (\Delta,1,0)$ and state $s = (\Delta,1,1)$ as $\pi_{0}(\Delta)$ and $\pi_{1}(\Delta)$, respectively. For virtual states, we define the corresponding steady state probability as $0$. We also define $\pi(\Delta)\triangleq \pi_{-1}(\Delta)+\pi_{0}(\Delta)+\pi_{1}(\Delta)$. Then, the expected AoII achieved by the strong preemptive policy is given by
\begin{equation}\label{eq-AoIISP}
\bar{\Delta}_{sp} = \sum_{\Delta=0}^{\infty}f(\Delta)\pi(\Delta).
\end{equation}
Hence, it is sufficient to calculate $\pi(\Delta)$ for $\Delta\geq0$.
\begin{lemma}\label{them-StationarySP}
The following gives the expressions of $\pi(\Delta)$ for $\Delta\geq0$.
\[
\pi(0)=\frac{p + q_1 - 2q_1p}{1-(1-q_1)(1-2p)}.
\]
For each $\Delta\geq1$,
\[
\pi(\Delta) = \frac{(1-q_1-p + 2q_1p)^{\Delta-1}(p^2 + q_1p - 2q_1p^2)}{1-(1-q_1)(1-2p)}.
\]
\end{lemma}
\begin{proof}
The balance equations the steady state probabilities satisfy can be obtained easily by exploiting the state transition probabilities detailed in Appendix~\ref{app-STP} of the supplementary material. Then, the stationary distribution can be obtained by solving the resulting system of linear equations. The complete proof can be found in Appendix~\ref{pf-StationarySP} of the supplementary material.
\end{proof}
\begin{remark}
We can approximate the infinite sum in \eqref{eq-AoIISP} using a finite sum with a sufficiently large upper bound on $\Delta$.
\end{remark}
We notice that
\begin{align*}
\bar{\Delta}_{sp} = & f(0)\pi(0) +c\sum_{\Delta=1}^{\infty}(1-q_1-p + 2q_1p)^{\Delta}f(\Delta)\\
< & f(0)\pi(0) +c\sum_{\Delta=0}^{\infty}(1-q_1-p + 2q_1p)^{\Delta}f(\Delta)\\
< & \infty,
\end{align*}
where $c \triangleq \frac{p^2 + q_1p - 2q_1p^2}{(1-q_1-p + 2q_1p)[1-(1-q_1)(1-2p)]}\in[0,1]$. Then, the finiteness of $\bar{\Delta}_{sp}$ is guaranteed by the assumption on $f(\Delta)$ introduced in Section~\ref{sec-AoII}. In the following, we provide the closed-form expression of the expected AoII achieved by the strong preemptive policy when $f(\Delta) = \alpha\Delta+\beta$, where $\alpha\geq0$ and $\beta\geq0$ are two finite constants.
\begin{corollary}\label{cor-SpecialPerformanceSP}
When $f(\Delta) = \alpha\Delta+\beta$,
\[
\bar{\Delta}_{sp} = \ddfrac{\alpha p}{(p+q_1-2q_1p)(q_1+2p-2q_1p)} + \beta.
\]
\end{corollary}
\begin{proof}
To deal with the infinite sum, we introduce an auxiliary quantity $\Sigma\triangleq\sum_{\Delta=2}^{\infty}\Delta\pi(\Delta)$. Then, $\bar{\Delta}_{sp} = \alpha(\pi(1) + \Sigma)+\beta$. To obtain the closed-form expression of $\Sigma$, we introduce another auxiliary quantity $\Pi \triangleq \sum_{\Delta=2}^{\infty}\pi(\Delta)$, whose closed-form expression can be obtained using Lemma~\ref{them-StationarySP}. The complete proof can be found in Appendix~\ref{pf-SpecialPerformanceSP} of the supplementary material.
\end{proof}

\subsection{Weak Preemptive Policy}
\begin{definition}[Weak preemptive policy]
The actions suggested by the weak preemptive policy are the same as those suggested by the strong preemptive policy except for the state $s$ with $\Delta>0$ and $i=1$, where the weak preemptive policy suggests staying idle.
\end{definition}
\begin{remark}
We notice that $i=1$ means that the transmitted update differs from the receiver's estimate. Thus, the update will bring new information to the receiver. Also, $\Delta > 0$ means that the receiver's current estimate is incorrect. Hence, the receiver needs a new update to adjust its estimate. The weak preemptive policy will let the transmission continue when $\Delta>0$ and $i=1$. Therefore, the weak preemptive policy is desirable for these two reasons.
\end{remark}
To ease the analysis, we consider the case where the transmission time is capped. More precisely, $q_{t_{max}} = 1$ where $t_{max}$ is the predetermined upper bound on the transmission time.
\begin{remark}
We can always choose a large enough $t_{max}$ such that the probability that an update takes more than $t_{max}$ time slots to be delivered is negligible.
\end{remark}
As we discussed in Section~\ref{sec-PerformanceSP}, the weak preemptive policy induces a DTMC with state space denoted by $\mathcal{S}_{wp}^{MC}$. Without loss of generality, we assume the system starts at state $(0,0,-1)$. Then, $\mathcal{S}_{wp}^{MC}$ consists of all the states that are accessible from state $(0,0,-1)$. For better presentations, we inherit the concept of virtual state introduced in Section~\ref{sec-PerformanceSP}. In the following, we elaborate on $\mathcal{S}_{wp}^{MC}$ and identify the virtual states. We note that when $q_{t_{max}} = 1$, an update can exist in the channel for at most $t_{max}$ time slots. Therefore, when the system is in state $s$ where $t=t_{max}-1$, the next state of the system must be the state $s'$ with $t'=0$ because the update will necessarily be delivered at the end of the $t_{max}$th time slot after the transmission starts or be preempted by the transmitter. Both cases result in the system going to state $s'$ with $t'=0$. Therefore, the system will never be in state $s$ with $t \geq t_{max}$. Then, $\mathcal{S}^{MC}_{wp}$ consists of the following states.
\begin{itemize}
\item $s=(\Delta,0,-1)$ where $\Delta\geq0$. The DTMC will be in this state every time the channel is idle.
\item $s=(\Delta,t,0)$ where $\Delta\geq0$ and $1\leq t\leq t_{max}-1$. Under the weak preemptive policy, the transmitter will preempt the update when $i=0$. Hence, the system will never reach the state with $t>1$ and $i=0$. Hence, $s=(\Delta,t,0)$ where $\Delta\geq0$ and $t>1$ is virtual state. Meanwhile, $i=0$ occurs only when the transmission starts when the AoII is zero. Hence, we know that $s=(\Delta,t,0)$ where $\Delta\geq2$ and $t=1$ is also virtual state.
\item $s=(\Delta,t,1)$ where $\Delta\geq0$ and $1\leq t\leq t_{max}-1$. $i=1$ occurs only when the transmission starts when the AoII is not zero. We recall that, within a single transition, the AoII either increases by $1$ or decreases to zero. Hence, $s=(1,t,1)$ where $1\leq t\leq t_{max}-1$ is virtual state.
\end{itemize}
We can easily conclude that all the states in the induced DTMC, except for the virtual states, communicate with the state $(0,0,-1)$. Hence, the stationary distribution of the induced DTMC exists. 
Let us denote by $\pi(\Delta,t,i)$ as the steady state probability of state $(\Delta,t,i)$. We define the steady state probability for the virtual state as $0$. We notice that the state transitions and the corresponding probabilities depend only on $\Delta$ when $a=1$. Hence, when calculating the stationary distribution, the states at which the weak preemptive policy suggests $a=1$ can be combined based on $\Delta$. Consequently, we define $\sum_{t,i}\pi(0,t,i)\triangleq\pi_0$ and $\pi(\Delta,0,-1) + \sum_{t=1}^{t_{max}-1}\pi(\Delta,t,0)\triangleq\pi_{\Delta}$. For state $s=(\Delta,t,1)$ where $\Delta>0$, we abbreviate the steady state probability as $\pi_\Delta(t)$. Then, the expected AoII achieved by the weak preemptive policy is given by
\[
\bar{\Delta}_{wp} = f(0)\pi_0 + \sum_{\Delta=1}^{\infty}f(\Delta)\left(\pi_\Delta+\sum_{t=1}^{t_{max}-1}\pi_{\Delta}(t)\right).
\]
We first solve for the stationary distribution. Combining with the system dynamics detailed in Appendix~\ref{app-STP} of the supplementary material, the steady state probabilities satisfy the following balance equations.
\begin{multline*}
\pi_0 = (1-p)\pi_0 + \bigg((1-q_1)p+q_1(1-p)\bigg)\sum_{\Delta=1}^{\infty}\pi_\Delta +\\
\sum_{t=1}^{t_{max}-1}\bigg[\bigg((1-q_{t+1})p+q_{t+1}(1-p)\bigg)\sum_{\Delta=1}^{\infty}\pi_\Delta(t)\bigg].
\end{multline*}
\[
\pi_1 = p\pi_0.
\]
\begin{equation}\label{eq-uniformperformance4}
\pi_\Delta = q_1p\pi_{\Delta-1}+\sum_{t=1}^{t_{max}-1}q_{t+1}p\pi_{\Delta-1}(t)\quad\Delta\geq2.
\end{equation}
\[
\pi_1(t) = 0\quad 1\leq t\leq t_{max}-1.
\]
\begin{equation}\label{eq-uniformperformance5}
\pi_\Delta(1) = (1-q_1)(1-p)\pi_{\Delta-1}\quad\Delta\geq2.
\end{equation}
\begin{multline}\label{eq-uniformperformance6}
\pi_{\Delta}(t) = (1-q_{t})(1-p)\pi_{\Delta-1}(t-1)\\
2\leq t\leq t_{max}-1\ and\ \Delta\geq2.
\end{multline}
\[
\pi_0 + \sum_{\Delta=1}^{\infty}\left(\pi_\Delta + \sum_{t=1}^{t_{max}-1}\pi_\Delta(t)\right) = 1.
\]
We notice that there are infinitely many balance equations. To overcome the infinity, we define $\Pi \triangleq \sum_{\Delta=1}^{\infty}\pi_\Delta$ and $\Pi(t) \triangleq \sum_{\Delta=1}^{\infty}\pi_\Delta(t)$. Leveraging the definitions, the balance equations can be rewritten as the following.
\begin{multline*}
\pi_0 = (1-p)\pi_0 + \bigg((1-q_1)p+q_1(1-p)\bigg)\Pi +\\
\sum_{t=1}^{t_{max}-1}\bigg[\bigg((1-q_{t+1})p+q_{t+1}(1-p)\bigg)\Pi(t)\bigg].
\end{multline*}
\[
\pi_1 = p\pi_0.
\]
\begin{equation}\label{eq-uniformperformance1}
\Pi - \pi_1 = q_1p\Pi + \sum_{t=1}^{t_{max}-1}q_{t+1}p\Pi(t).
\end{equation}
\begin{equation}\label{eq-uniformperformance2}
\Pi(1) = (1-q_1)(1-p)\Pi.
\end{equation}
\begin{equation}\label{eq-uniformperformance3}
\Pi(t) = (1-q_{t})(1-p)\Pi(t-1)\quad 2\leq t\leq t_{max}-1.
\end{equation}
\begin{equation}\label{eq-uniformperformance7}
\pi_0 + \Pi + \sum_{t=1}^{t_{max}-1}\Pi(t) = 1.
\end{equation}
Note that \eqref{eq-uniformperformance1}, \eqref{eq-uniformperformance2}, and \eqref{eq-uniformperformance3} is obtained by summing \eqref{eq-uniformperformance4}, \eqref{eq-uniformperformance5}, and \eqref{eq-uniformperformance6} over $\Delta$ from $2$ to $\infty$, respectively. Then, we solve the above system of linear equations.
\begin{lemma}\label{lem-StationaryWP}
The closed-form expression of $\Pi$ is given by \eqref{eq-EquivalentEq1}
\begin{figure*}[!t]
\normalsize
\begin{equation}\label{eq-EquivalentEq1}
\Pi = \ddfrac{1}{\frac{1}{p}-q_1 - \left[\sum_{t=1}^{t_{max}-1}q_{t+1}\left(\prod_{l=1}^t\mathcal{P}_l\right)\right] + 1 + \sum_{t=1}^{t_{max}-1}\left(\prod_{l=1}^t\mathcal{P}_l\right)}.
\end{equation}
\hrulefill
\vspace*{4pt}
\end{figure*}
and
\[
\Pi(t) = \prod_{l=1}^t\mathcal{P}_l\Pi \quad 1\leq t\leq t_{max}-1,
\]
where
\[
\mathcal{P}_t = (1-q_t)(1-p).
\]
\end{lemma}
\begin{proof}
The complete proof can be found in Appendix~\ref{pf-StationaryWP} of the supplementary material.
\end{proof}
 In the following, we consider $f(\Delta) = \alpha\Delta+\beta$ to facilitate the analysis. As we will see later, $\Pi$ and $\Pi(t)$ for $1\leq t\leq t_{max}-1$ are sufficient to calculate $\bar{\Delta}_{wp}$ when $f(\Delta) = \alpha\Delta+\beta$. In this case,
\[
\bar{\Delta}_{wp} = \alpha\sum_{\Delta=1}^{\infty}\left(\Delta\left[\pi_\Delta+\sum_{t=1}^{t_{max}-1}\pi_\Delta(t)\right]\right)+\beta.
\]
Similar to what we did in Corollary~\ref{cor-SpecialPerformanceSP}, we define $\Sigma \triangleq \sum_{\Delta=1}^{\infty}\Delta\pi_\Delta$ and $\Sigma(t) \triangleq \sum_{\Delta=1}^{\infty}\Delta\pi_\Delta(t)$ for $1\leq t\leq t_{max}-1$ to avoid the infinite sum. Then,
\begin{equation}\label{eq-AoIIwp}
\bar{\Delta}_{wp} = \alpha\left(\Sigma + \sum_{t=1}^{t_{max}-1}\Sigma(t)\right)+\beta.
\end{equation}
Hence, it is sufficient to obtain the closed-form expressions of $\Sigma$ and $\Sigma(t)$ for $1\leq t\leq t_{max}-1$.
\begin{theorem}\label{thm-SpecialPerformanceWP}
When $f(\Delta) = \alpha\Delta+\beta$,
\[
\bar{\Delta}_{wp} =\alpha\left(\Sigma + \sum_{t=1}^{t_{max}-1}\Sigma(t)\right)+\beta,
\]
where
\[
\Sigma = \ddfrac{\Pi + \sum_{t=1}^{t_{max}-1}\left\{p_{t+1}p\left[\sum_{i=1}^t\left(\prod_{j=i+1}^t\mathcal{P}_j\right)\Pi(i)\right]\right\}}{1-p_1p-\sum_{t=1}^{t_{max}-1}\left[p_{t+1}p\left(\prod_{l=1}^t\mathcal{P}_l\right)\right]},
\]
and for each $\leq t\leq t_{max}-1$,
\[
\Sigma(t) = \left(\prod_{l=1}^{t}\mathcal{P}_l\right)\Sigma + \sum_{i=1}^{t}\left[\left(\prod_{j=i+1}^t
\mathcal{P}_j\right)\Pi(i)\right].
\]
\end{theorem}
\begin{proof}
The closed-form expression of $\Sigma$ and $\Sigma(t)$ for $1\leq t\leq t_{max}-1$ can be obtained using $\Pi$ and $\Pi(t)$ for $1\leq t\leq t_{max}-1$. The complete proof can be found in Appendix~\ref{pf-SpecialPerformanceWP} of the supplementary material.
\end{proof}

\section{Optimal Policy}\label{sec-OptimalPolicy}
In this section, we first prove the existence of the optimal policy. Then, we provide a feasible relative value iterative algorithm that approximates the optimal policy. Next, using the policy improvement theorem, we analyze the optimization problem \eqref{eq-Goal} and theoretically find the optimal policy when the transmission delay follows the Geometric distribution and the Zipf distribution, respectively.

\subsection{Existence of the Optimal Policy}
For the $\mathcal{M}$ in Section~\ref{sec-MDP}, we define the expected $\gamma$-discounted cost under policy $\psi$ as
\[
V_{\psi,\gamma}(s) = \mathbb{E}_{\psi}\left[\sum_{k=0}^{\infty}\gamma^tC(s_k)\mid s\right],
\]
where $0 < \gamma < 1$ is a discount factor and $s_k$ is the state of $\mathcal{M}$ at time $k$. Let $V_{\gamma}(s)$ be the value function associated with $\mathcal{M}$ under $\gamma$-discounted cost. Then, we know that $V_{\gamma}(s) = inf_{\psi} V_{\psi,\gamma}(s)$. Moreover, $V_{\gamma}(s)$ satisfies the Bellman equation.
\[
V_{\gamma}(s) = \min_{a\in\mathcal{A}}\left\lbrace C(s)+\gamma\sum_{s'\in\mathcal{S}}P_{s,s'}(a)V_{\gamma}(s')\right\rbrace\quad s\in\mathcal{S}.
\]
The value iteration algorithm is one of the most commonly used algorithms to calculate the value function. Let $V_{\gamma,\nu}(s)$ be the estimated value function at iteration $\nu$. Then, the estimated value function is updated in the following way.
\begin{equation}\label{eq-ValueIterationGammaDiscount}
V_{\gamma,\nu+1}(s) = \min_{a\in\mathcal{A}}\left\lbrace C(s)+\gamma\sum_{s'\in\mathcal{S}}P_{s,s'}(a)V_{\gamma,\nu}(s')\right\rbrace\quad s\in\mathcal{S}.
\end{equation}
Without loss of generality, we initialize $V_{\gamma,0}(s)=0$ for $s\in\mathcal{S}$. Then, we can prove the following lemma.
\begin{lemma}\label{lem-Converge}
When updated following \eqref{eq-ValueIterationGammaDiscount}, $\lim_{\nu\rightarrow\infty}V_{\gamma,\nu}(s)= V_{\gamma}(s)$ for $s\in\mathcal{S}$.
\end{lemma}
\begin{proof}
According to~\cite[Propositions 1 and 3]{b16}, it is sufficient to show that $V_{\gamma}(s)$ is finite. To this end, we have
\begin{align*}
V_{\psi,\gamma}(s) = & \mathbb{E}_{\psi}\left[\sum_{k=0}^{\infty}\gamma^kC(s_k)\ |\ s\right]\\
\leq & \sum_{k=0}^{\infty}\gamma^k f(\Delta+k) =\ddfrac{1}{\gamma^{\Delta}}\sum_{k=\Delta}^{\infty}\gamma^{k} f(k) \\
\leq & \frac{1}{\gamma^{\Delta}}\sum_{k=0}^{\infty}\gamma^kf(k)< \infty.
\end{align*}
The finiteness is guaranteed by the assumption on $f(\Delta)$ introduced in Section~\ref{sec-AoII}. Then, by definition, we have $V_{\gamma}(s)\leq V_{\gamma,\psi}(s)<\infty$. Hence, we can conclude that the value iteration reported in \eqref{eq-ValueIterationGammaDiscount} will converge to the value function.
\end{proof}
Leveraging the iterative nature of the value iteration algorithm, we can prove the following structural property of $V_{\gamma}(s)$.
\begin{lemma}\label{lem-Monotone}
$V_{\gamma}(s)$ is non-decreasing in $\Delta>0$.
\end{lemma}
\begin{proof}
Given the convergence proved in Lemma~\ref{lem-Converge}, the monotonicity of $V_{\gamma}(s)$ can be proved via mathematical induction. The complete proof can be found in Appendix~\ref{pf-Monotone} of the supplementary material.
\end{proof}
Now, we proceed with showing the existence of the optimal policy. To this end, we first define the stationary policy.
\begin{definition}[Stationary policy]
A stationary policy specifies a single action in each time slot.
\end{definition}
\begin{theorem}\label{thm-optimalexit}
There exists a stationary policy $\psi$ that is optimal for $\mathcal{M}$. Moreover, the minimum expected AoII is independent of the initial state.
\end{theorem}
\begin{proof}
The proof follows the same steps as in~\cite[Theorem 4]{b11}. We define $h_{\gamma}(s)\triangleq V_{\gamma}(s)- V_{\gamma}(s^{ref})$ as the relative value function, where $s^{ref}$ is the reference state. Note that the reference state is arbitrary but fixed. Then, we show that $\mathcal{M}$ satisfies the two conditions given in~\cite{b16}. To avoid excessive repetition of the proof, we omit the specific reasoning and give only the proofs of two important intermediate results.
\begin{enumerate}
\item $h_{\gamma}(s)$ is non-decreasing in $\Delta$ when $\Delta>0$. The result follows directly from Lemma~\ref{lem-Monotone} because the reference state is fixed.
\item There exists a policy $\psi$ that induces an irreducible ergodic Markov chain, and the expected cost is finite. $\psi$ can be the strong preemptive policy. Then, the result is true as discussed in Section~\ref{sec-PerformanceSP}.
\end{enumerate}
Using these two results, we can verify that $\mathcal{M}$ satisfies the two conditions given in~\cite{b16}. Then, the existence of the optimal policy is guaranteed by~\cite[Theorem]{b16}. Moreover, the minimum expected cost is independent of the initial state.
\end{proof}

\subsection{Value Iteration Algorithm}\label{sec-VIA}
In this section, we present the relative value iteration (RVI) algorithm that approximates the optimal policy for $\mathcal{M}$. Direct application of RVI becomes impractical as the state space $\mathcal{S}$ of $\mathcal{M}$ is infinite. Hence, we construct another $\mathcal{M}^{(m)} = (\mathcal{S}^{(m)},\mathcal{A},\mathcal{P}^{(m)},\mathcal{C})$ by truncating the value of $\Delta$ and $t$. More precisely, we impose
\[
{\mathcal{S}^{(m)}}:
\begin{dcases}
& \Delta \in \{0,1,...,\Delta_{max}\},\\
& i \in \{-1,0,1\},\\
& t \in \{0,1,...,t_{max}\},
\end{dcases}
\]
where $\Delta_{max}$ and $t_{max}$ are the predetermined maximal value of $\Delta$ and $t$, respectively. Then, the size of the state space reduces from infinite to $((\Delta_{max}+1)\times 3\times (t_{max}+1))$. The transition probabilities from $s\in\mathcal{S}^{(m)}$ to $z\in\mathcal{S}\setminus\mathcal{S}^{(m)}$ are redistributed to the states $s'\in\mathcal{S}^{(m)}$ according to \eqref{eq-redistribut}
\begin{figure*}[!t]
\normalsize
\begin{equation}\label{eq-redistribut}
P^{(m)}_{s,s'}(a) = \begin{cases}
          P_{s,s'}(a) & \Delta'<\Delta_{max}\ and\ t'<t_{max}, \\
          P_{s,s'}(a) + \sum_{G_1(z,s')}P_{s,z}(a) & \Delta'=\Delta_{max}\ and\ t'<t_{max}, \\
          P_{s,s'}(a) + \sum_{G_2(z,s')}P_{s,z}(a) & \Delta'<\Delta_{max}\ and\ t'=t_{max}, \\
          P_{s,s'}(a) + \sum_{G_3(z,s')}P_{s,z}(a) & \Delta'=\Delta_{max}\ and\ t'=t_{max}.
     \end{cases}
\end{equation}
\hrulefill
\vspace*{4pt}
\end{figure*}
where $G_1(s,s') = \{s:\Delta>\Delta_{max},t=t',i=i',\}$, $G_2(s,s') = \{s:\Delta=\Delta',t>t_{max},i=i'\}$, and $G_3(s,s') = \{s:\Delta>\Delta_{max},t>t_{max},i=i'\}$. The action space $\mathcal{A}$ and the immediate cost $\mathcal{C}$ are the same as defined in $\mathcal{M}$.

We can rigorously show that the sequence of optimal policies for $\mathcal{M}^{(m)}$ will converge to the optimal policy for $\mathcal{M}$ as $\Delta_{max}\rightarrow\infty$ and $t_{max}\rightarrow\infty$. More specifically, we can show that our system verifies the two assumptions given in~\cite{b14}. Then, by~\cite[Theorem 2.2]{b14}, we know the following results hold.
\begin{itemize}
\item There exists an average cost optimal stationary policy for $\mathcal{M}^{(m)}$.
\item Any limit point of the sequence of optimal policies for $\mathcal{M}^{(m)}$ is optimal for $\mathcal{M}$.
\end{itemize}
Considering the similarity of the system model, the proof will be similar to the proof of~\cite[Theorem 1]{b6}. Therefore, we only prove one of the most important lemmas in the proof and omit the rest of the proof. Let $V^{(m)}_{\gamma}(s)$ be the value function associated with $\mathcal{M}^{(m)}$ under $\gamma$-discounted cost.
\begin{lemma}\label{lem-MonotoneM}
$V^{(m)}_{\gamma}(s)$ is non-decreasing in $\Delta>0$.
\end{lemma}
\begin{proof}
The proof is very similar to the proof of Lemma~\ref{lem-Monotone} because the way of redistributing the state transition probabilities shown in \eqref{eq-redistribut} does not change the structural properties of the state transition probability presented in the proof of Lemma~\ref{lem-Monotone}. Therefore, we omit the proof of this lemma.
\end{proof}
Then, we can apply RVI to the truncated MDP $\mathcal{M}^{(m)}$ and treat the resulting optimal policy as an approximation of the optimal policy for $\mathcal{M}$. The pseudocode of RVI is given in Algorithm~\ref{alg-RVIA}.
\begin{algorithm}[!t]
    \begin{algorithmic}[1]
    \Procedure{RVI}{$\mathcal{M}^{(m)}$,$\epsilon$}
        \State $\nu\gets0$; $V_{\nu}(s)\gets0$ for $s\in\mathcal{S}^{(m)}$
        \State Choose $s^{ref}\in\mathcal{S}^{(m)}$ arbitrarily
        \Repeat
            \For{$s \in \mathcal{S}^{(m)}$}
                \For{$a \in \mathcal{A}$}
                    \State $Q_{a}(s) \gets C(s) + \sum_{s'} P^{(m)}_{s,s'}(a) V_{\nu}(s')$
                \EndFor
            	\State $Q(s) \gets \min_a \{Q_{a}(s)\}$
            \EndFor
            \For{$s \in \mathcal{S}^{(m)}$}
                \State $V_{\nu+1}(s) \gets Q(s) - Q(s^{ref})$
            \EndFor
            \State $\nu \gets \nu + 1$
        \Until{$\max_s\{\left|V_{\nu}(s)-V_{\nu-1}(s)\right|\}\leq\epsilon$}
        \State \Return $\hat{\psi}^* \gets \argmin_a \{Q_{a}(s)\}$
    \EndProcedure
    \end{algorithmic}
\caption{Relative Value Iteration}
\label{alg-RVIA}
\end{algorithm}
A similar approximation is also used in~\cite{b9}, according to which small $\Delta_{max}$ and $t_{max}$ can give an accurate estimate of the optimal policy for $\mathcal{M}$. However, the choices of $\Delta_{max}$ and $t_{max}$ are still problematic. If large values are chosen, the state space of $\mathcal{M}^{(m)}$ grows rapidly. Meanwhile, the RVI may result in a non-optimal policy if the chosen value is small. Hence, in the following subsections, we investigate two specific examples of the transmission delay, namely the transmission delay that follows the Geometric distribution and the Zipf distribution. For these two common delay models, we theoretically find their optimal policies. For this purpose, we first introduce the policy iteration algorithm and the policy improvement theorem.

\subsection{Policy Iteration Algorithm}
The pseudocode of policy iteration algorithm is given in Algorithm \ref{alg-PIA}.
\begin{algorithm}[!t]
    \begin{algorithmic}[1]
    \Procedure{PI}{$\mathcal{M}$}
        \State Choose $\psi'(s)\in\mathcal{A}$ for $s\in\mathcal{S}$ arbitrarily
        \Repeat
        \State $\psi\gets\psi'$
        \State $V^{\psi}(s)\leftarrow$ \textsc{PolicyEvaluation}$(\mathcal{M},\psi)$
        \State $\psi'\leftarrow$ \textsc{PolicyImprovement}$(\mathcal{M},V^{\psi}(s))$
       	\Until{$\psi'=\psi$}
        \State\Return $\psi^*\gets \psi$
    \EndProcedure
    \end{algorithmic}
\caption{Policy Iteration}
\label{alg-PIA}
\end{algorithm}
We elaborate on the \textsc{PolicyEvaluation} function and the \textsc{PolicyImprovement} function.
\begin{itemize}
\item The \textsc{PolicyEvaluation} function takes the MDP $\mathcal{M}$ and the policy $\psi$ as input and produce the value function $V^{\psi}(s)$ and the expected AoII $\theta^{\psi}$ resulting from the adoption of $\psi$. To be more specific, $V^{\psi}(s)$ and $\theta^{\psi}$ are the solution to the following system of linear equations.
\begin{equation}\label{eq-PolicyIteration1}
V^{\psi}(s) + \theta^{\psi} = C(s) + \sum_{s'\in\mathcal{S}}P_{s,s'}^{\psi}V^{\psi}(s')\quad s\in\mathcal{S},
\end{equation}
where $P_{s,s'}^{\psi}$ is the probability that the system will transit from state $s$ to state $s'$ under policy $\psi$. Note that \eqref{eq-PolicyIteration1} forms a underdetermined system. Hence, we can select a reference state $s^{ref}$ arbitrarily and set $V^{\psi}(s^{ref}) = 0$. In this way, we can obtain a unique solution.
\item The \textsc{PolicyImprovement} function takes the MDP $\mathcal{M}$ and the value function $V^{\psi}(s)$ as input and produce the optimal policy $\psi'$ under $V^{\psi}(s)$. Let $\psi'(s)$ be the action suggested by the new policy $\psi'$ at state $s$. Then, $\psi'(s)$ is given by the following equation.
\[
\psi'(s) = \argmin_{a\in\mathcal{A}}\left\{C(s) + \sum_{s'\in\mathcal{S}}P_{s,s'}(a)V^{\psi}(s')\right\}.
\]
\end{itemize}
The policy iteration algorithm iterates between the two functions until convergence. The convergence criterion is defined at line $7$ of Algorithm \ref{alg-PIA}. Although the policy iteration algorithm appears to be more computationally demanding than the value iteration algorithm and is not as commonly used as the value iteration algorithm, it has the advantage that we can prove the policy improvement theorem, which will be the basis of our theoretical analysis in the next two subsections.
\begin{theorem}[Policy improvement theorem]\label{thm-policyimprovement}
Suppose that we have obtained the value function resulting from the operation of a policy $A$ and that the policy improvement function has produced a new policy $A'$. When policy $A$ and policy $A'$ are identical, we say the policy improvement function converges and policy $A$ is optimal.
\end{theorem}
\begin{proof}
The proof is based on~\cite[pp.42-43]{b18}. We first assume that the policy improvement function converges to a non-optimal policy $A$. Then, we prove that there is a contradiction under this assumption. Thus, the assumption that $A$ is non-optimal must be false, and its opposite must be true. The complete proof can be found in Appendix~\ref{pf-policyimprovement} of the supplementary material.
\end{proof}
With the policy iteration algorithm and Theorem~\ref{thm-policyimprovement} in mind, we can proceed with finding the optimal policy through theoretical analysis.

\subsection{Geometric Delay}
In this subsection, we consider the case where the transmission delay is Geometrically distributed with success probability $0<p_s<1$. More precisely, 
\[
p_t = p_s(1-p_s)^{t-1}\quad t\geq 1.
\]
\begin{remark}\label{rm-ps}
We omit the case of $p_s=0$ as, in this case, the update will never be delivered. We also do not discuss the case of $p_s=1$ because the transmission time is deterministic and normalized in this case. The corresponding optimal policy has been discussed in many papers~\cite{b4,b5,b6}.
\end{remark}
Under Geometric distribution, $q_t = p_s$ for $t\geq1$. Then, leveraging the policy improvement theorem, we can prove the following theorem.
\begin{theorem}\label{thm-optimalgeometric}
The strong preemptive policy is optimal if the transmission delay follows the Geometric distribution.
\end{theorem}
\begin{proof}
According to Theorem~\ref{thm-policyimprovement}, it is sufficient to prove that the policy iteration algorithm converges to the strong preemptive policy. Specifically, we first calculate the value function resulting from the strong preemption policy. Then we use the resulting value function to derive a new policy and verify that the old and new policies are the same. The complete proof can be found in Appendix~\ref{pf-optimalgeometric} of the supplementary material.
\end{proof}
\begin{remark}
The optimality of the strong preemptive policy is intuitive since $q_t$ is independent of $t$, which means that the probability of an update being delivered is independent of how long it has been in transmission. Thus, the strong preemption policy ensures that updates in the channel will always be the latest ones without sacrificing the likelihood of their delivery.
\end{remark}

\subsection{Zipf Delay}
In this subsection, we consider the case where the transmission delay follows the Zipf distribution with the constant $a$. More precisely,
\[
p_t = \frac{t^{-a}}{\sum_{i=1}^{t_{max}}i^{-a}}\quad 1\leq t\leq t_{max},
\]
where $t_{max}>1$ is a predetermined constant. 
\begin{remark}
When $t_{max}=1$, the transmission time is deterministic and normalized. Hence, we omit the discussion on this case for the same reason detailed in Remark~\ref{rm-ps}.
\end{remark}
A transmission delay that follows the Zipf distribution is also considered in the literature on information freshness~\cite{b9,b19}. Under the Zipf distribution,
\[
q_t = \frac{t^{-a}}{\sum_{i=t}^{t_{max}}i^{-a}}\quad 1\leq t\leq t_{max}.
\]
We define that $q_t\triangleq 0$ when $t>t_{max}$. Hence, the transmission time of an update is upper bounded by $t_{max}$. Consequently, the state $s=(\Delta,t,i)$ in the corresponding $\mathcal{M}$ satisfies $0\leq t\leq t_{max}-1$. To simplify the analysis, we consider the case of $f(\Delta)=\alpha\Delta+\beta$. Before the optimal policy, we first introduce the threshold preemptive policy and evaluate its performance.
\begin{definition}[Threshold preemptive policy]
The threshold preemptive policy always starts a new transmission when the channel is idle and does not preempt updates only at state $s=(\Delta,t_{max}-1,1)$ where $\Delta\geq1$.
\end{definition}
The following theorem gives the expected AoII achieved by the threshold preemptive policy when $f(\Delta)=\alpha\Delta+\beta$.
\begin{theorem}\label{prop-weakpreemptiveperformance}
When $f(\Delta)=\alpha\Delta+\beta$, the expected AoII achieved by the threshold preemptive policy $\bar{\Delta}_{tp}$ is give by
\[
\bar{\Delta}_{tp} = \frac{\alpha p}{(p+q_1-2q_1p)(q_1+2p-2q_1p)} + \beta.
\]
\end{theorem}
\begin{proof}
Although the threshold preemptive policy and the strong preemptive policy are not precisely the same, they yield the same expected AoII. This is because the actions suggested by the two policies differ only in the virtual states, which does not affect the long-term average performance.
\end{proof}
We first introduce the following condition for direct use in the subsequent theoretical analysis.
\begin{condition}\label{con}
The conditions are the following.
\begin{itemize}
\item $q_1\geq q_t$ for $1\leq t\leq t_{max}-2$.
\item When $t_{max}\geq3$, $\mathcal{Q}_1\geq0$, $\mathcal{Q}_2\geq0$, and $\mathcal{Q}_3\geq0$, where $\mathcal{Q}_1$, $\mathcal{Q}_2$, and $\mathcal{Q}_3$ are given by \eqref{eq-EquivalentEq2}.
\begin{figure*}[!t]
\normalsize
\begin{equation}\label{eq-EquivalentEq2}
\begin{split}
\mathcal{Q}_1 & \triangleq (q_{t_{max}-1}-q_{t_{max}-1}p-p)+(1-q_{t_{max}-1})p(q_1+2p-2q_1p)^2.\\
\mathcal{Q}_2 & \triangleq\frac{(1-2p)\{(q_1-1)+(1-q_{t_{max}-1})[p+q_1(1-p)]\}}{q_1+p-2q_1p}.\\
\mathcal{Q}_3 & \triangleq \frac{(1-q_1)(2p-1)-p(1-q_{t_{max}-1})}{(2p+q_1-2q_1p)(p+q_1-2q_1p)} +  \frac{(1-q_{t_{max}-1})(1-p)p}{q_1+p-2q_1p}+(1-q_{t_{max}-1})(1-p) + \mathcal{Q}_2. 
\end{split}
\end{equation}
\hrulefill
\vspace*{4pt}
\end{figure*}
\end{itemize}
\end{condition}
Then, we can prove the following theorem.
\begin{theorem}\label{thm-optimalzipf}
When $f(\Delta) = \alpha\Delta+\beta$ and under Condition~\ref{con}, the threshold preemptive policy is optimal if the transmission delay follows the Zipf distribution.
\end{theorem}
\begin{proof}
We follow the same methodology presented in the proof of Theorem~\ref{thm-optimalgeometric}. The complete proof can be found in Appendix~\ref{proof-optimalzipf} of the supplementary material.
\end{proof}
\begin{remark}
For the system that fails to satisfy Condition~\ref{con}, we can use the relative value iteration algorithm introduced in Section~\ref{sec-VIA} to approximate the corresponding optimal policy. 
\end{remark}
\begin{corollary}
The following results can be derived from Theorem~\ref{thm-optimalzipf}.
\begin{enumerate}
\item When $f(\Delta) = \alpha\Delta+\beta$ and under Condition~\ref{con}, the threshold preemptive policy is optimal.
\item For a generic transmission delay with an upper bound of $2$ time slots, the threshold preemptive policy is optimal.
\end{enumerate}
\end{corollary}
\begin{proof}
We note that in the proof of Theorem~\ref{thm-optimalzipf}, we only use Condition~\ref{con} and the fact that $q_t\geq0$. Therefore, the first result can be directly derived from the proof of Theorem~\ref{thm-optimalzipf}. For the second result, since the transmission time is upper bounded by $2$, the proof follows the same steps as detailed in the proof of Theorem~\ref{thm-optimalzipf} with $t_{max}=2$. The difference is that only the first three structural properties in Lemma~\ref{lem-zipfvfproperty} hold. Nevertheless, we can still complete the proof because the case that needs to use the fourth structural property in Lemma~\ref{lem-zipfvfproperty} does not exist in the case of $t_{max}=2$. For the same reason, we also do not need to verify Condition~\ref{con}. Consequently, we omit the detailed proof.
\end{proof}

\section{Numerical Results}\label{sec-Numerical}
In this section, we present numerical results regarding the verification of Condition~\ref{con} as well as a performance analysis of the optimal policy and the performance improvement compared to the non-preemptive policy.

\subsection{Condition~\ref{con} Verification}
In this subsection, we numerically verify Condition~\ref{con} under various system parameters. More specifically, the system parameters are chosen as follows.
\begin{itemize}
\item $0.05\leq p\leq 0.45$ with an increment of $0.05$.
\item $0\leq a\leq 5$ with an increment of $0.25$.
\item $3\leq t_{max}\leq 11$ with an increment of $1$.
\end{itemize}
We choose $f(\Delta) = \Delta$ for better illustration. The results are summarized in Table~\ref{tab-signs}, where the cross means that Condition~\ref{con} is not satisfied, the check mark means that Condition~\ref{con} is satisfied, and the circle means the result depends on the specific parameters.
\begin{table}[!t]
\caption{Condition~\ref{con} Check}
\label{tab-signs}
\begin{tabularx}{\columnwidth}{@{} l *{7}{C} c @{}}
\toprule
$a$ & $0$ & $0.25$ & $0,5$ & $0.75$ 
& $1$ & $1.25$ & $1.5$ \\ 
\midrule
Result  & \Cross  & \Cross  & \Cross  & \Cross  & \Cross  & \Cross   & \Cross  \\ 
\toprule
$a$ & $1.75$ & $2$ & $2.25$ & $2.5$ & $2.75$ & $3$ & $3.25$ \\ 
\midrule
Result & \Cross & \Cross & \Circle & \Checkmark  & \Checkmark  & \Checkmark  & \Checkmark  \\ 
\toprule
$a$ & $3.5$ & $3.75$ & $4$ & $4.25$ 
& $4.5$ & $4.75$ & $5$ \\ 
\midrule
Result  & \Checkmark  & \Checkmark  & \Checkmark  & \Checkmark  & \Checkmark  & \Checkmark   & \Checkmark \\ 
\bottomrule
\end{tabularx}
\end{table}
When $a=2.25$, the results are visualized in Fig.~\ref{fig-ZipfConCheck}.
\begin{figure}[!t]
\centering
\includegraphics[width=3in]{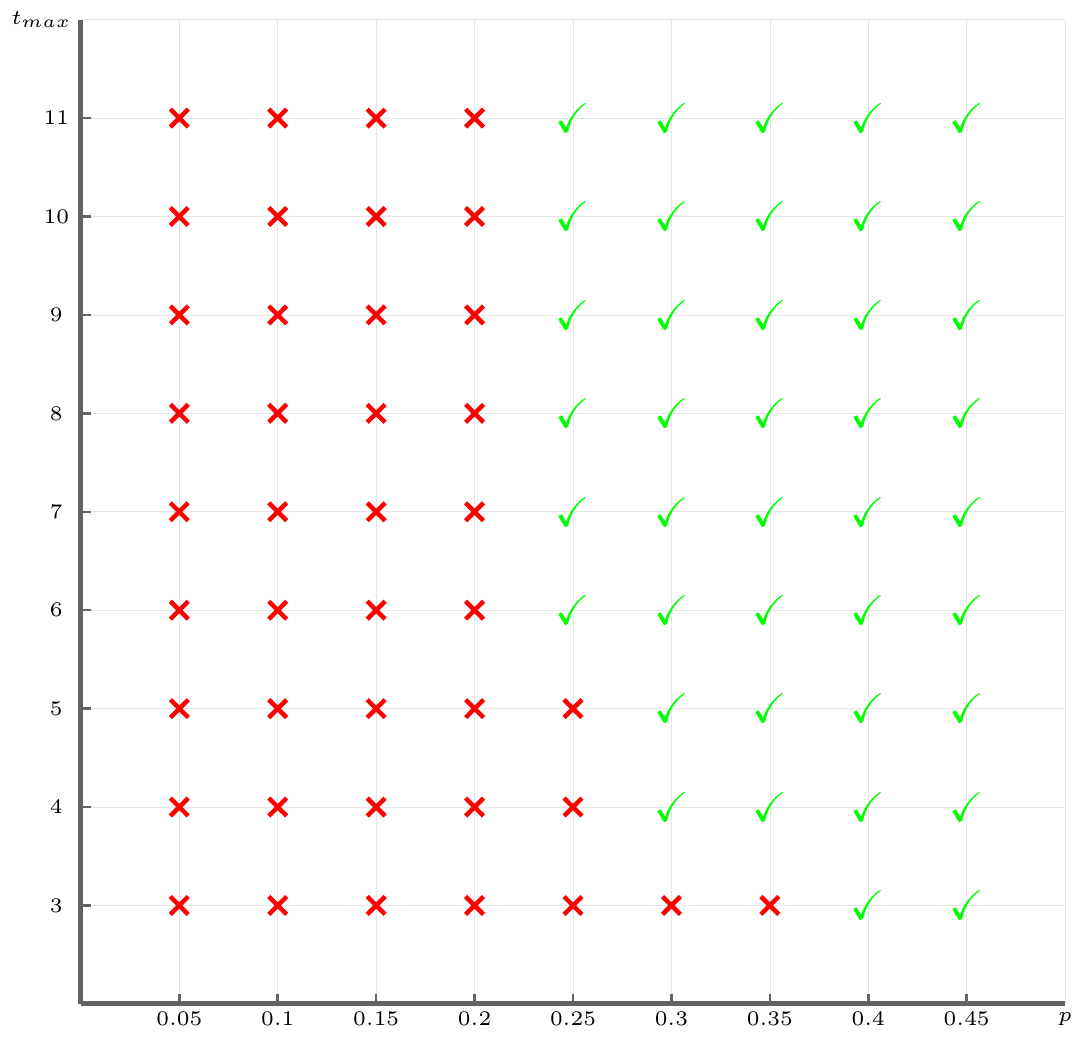}
\caption{A visual representation of the results of numerical check of Condition~\ref{con} when the transmission delay follows the Zipf distribution with $a = 2.25$ under different $t_{max}$ and $p$. In the figure, the check mark indicates that Condition~\ref{con} is verified, and the cross indicates that Condition~\ref{con} is not verified.}
\label{fig-ZipfConCheck}
\end{figure}
We emphasize here that the optimal policy depends not only on the type of probability distribution of the delay but also on the probability distribution parameters.

\subsection{Performance of the Optimal Policy}
In this subsection, we compare the performance of the optimal policy with the non-preemptive policy to highlight the performance improvements brought about by the preemption capability. To this end, we first define a specific type of non-preemptive policy.
\begin{definition}[Threshold policy]
The threshold policy starts a new transmission when the channel is idle, and the AoII is not zero. When the channel is busy, the threshold policy never preempts the transmitting update.
\end{definition}
We choose $f(\Delta)=\Delta$ for better illustrations. Note that the threshold policy is a special case of the threshold policy defined in~\cite[Definition 2]{b11}. Then, we can compute the performances of the optimal policy and the threshold policy using Corollary~\ref{cor-SpecialPerformanceSP} and~\cite[Theorem 3]{b11}, respectively. To accommodate assumption 1 in~\cite{b11}, we set the upper bound on the transmission time to $40$ when calculating the performance of the threshold policy. We also choose the system parameters that have been verified in~\cite{b11} to satisfy~\cite[Condition 1]{b11}. Consequently, the threshold policy is optimal when the transmitter has no preemption capability. Then, we plot the corresponding performances for the two typical transmission delay models studied in this paper.

When the transmission delay follows the Geometric distribution, the numerical results are given in Fig.~\ref{fig-PerformanceGeo}.
\begin{figure}[!t]
\centering
\begin{subfigure}{3in}
\centering
\includegraphics[width=3in]{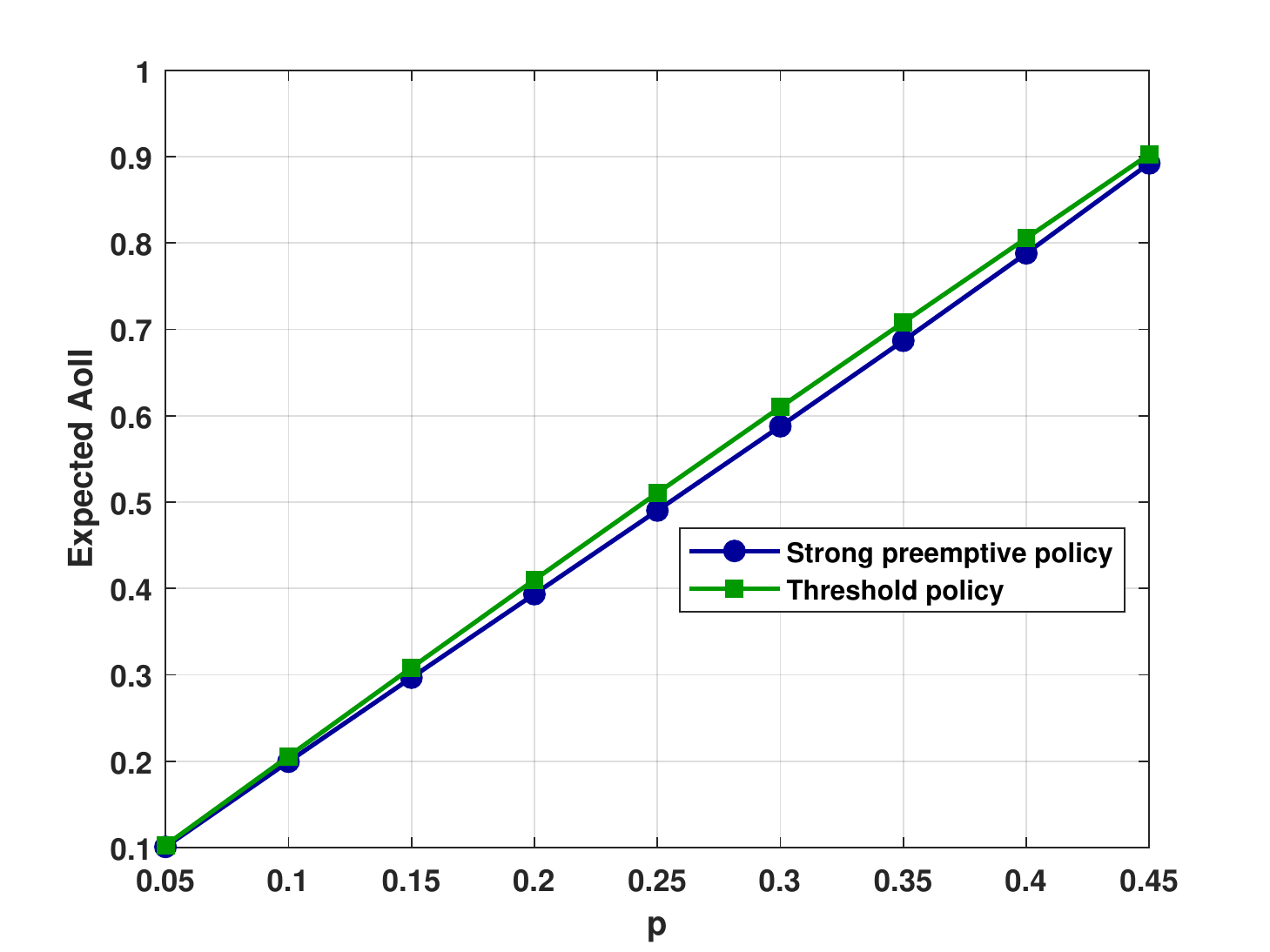}
\caption{When $p_s=0.7$.}
\label{fig-PerformanceGeop}
\end{subfigure}\hfill
\begin{subfigure}{3in}
\centering
\includegraphics[width=3in]{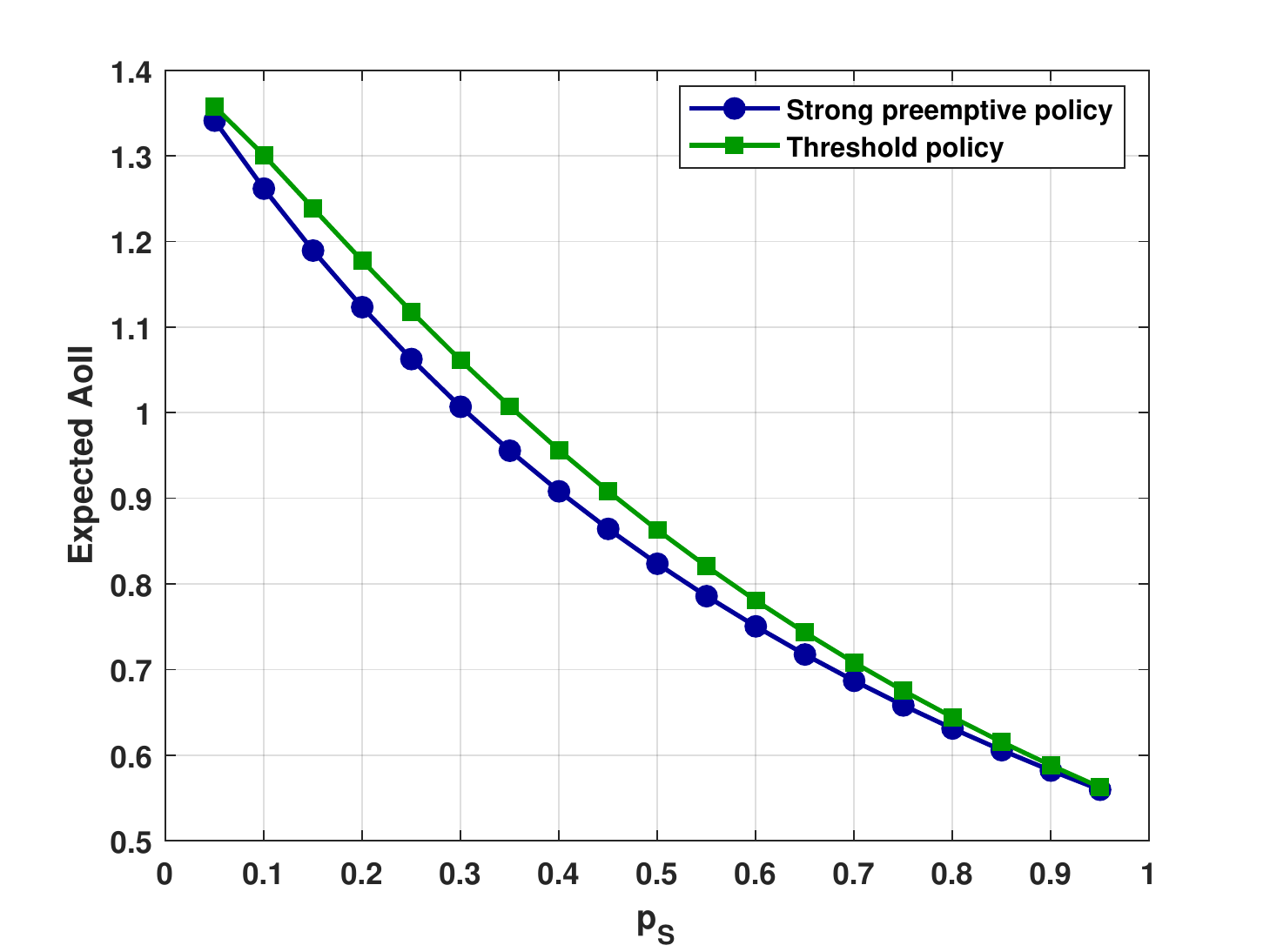}
\caption{When $p=0.35$.}
\label{fig-PerformanceGeops}
\end{subfigure}\hfill
\caption{The performance comparison when the transmission delay follows the Geometric distribution. In this case, there are two system parameters. One is the Markovina source dynamics $p$, and the other is the success probability $p_s$ in the Geometric distribution. Therefore, we fix one of the parameters and plot the corresponding results when the other parameter varies.}
\label{fig-PerformanceGeo}
\end{figure}
The plots show that the performance gain from the transmitter's preemption capability is not significant. One possible reason is that the source process modeling and the time penalty function choice in this paper are simple. Meanwhile, the expected AoII achieved by the optimal policy increases as $p$ increases. This is because when $p$ is large, the Markovian source jumps between states more frequently, making it more difficult for the receiver to maintain a correct estimate about the state of the Markovian source. On the contrary, the expected AoII resulting from the optimal policy decreases as $p_s$ increases. The reason for this is as follows. When $p_s$ is large, the expected transmission time of an update is small. As a result, the receiver receives more updates per unit of time, which allows a more accurate estimation of the state of the Markovian source.

When the transmission delay follows the Zipf distribution, the numerical results are given in Fig.~\ref{fig-PerformanceZipf}.
\begin{figure}[!t]
\centering
\begin{subfigure}{3in}
\centering
\includegraphics[width=3in]{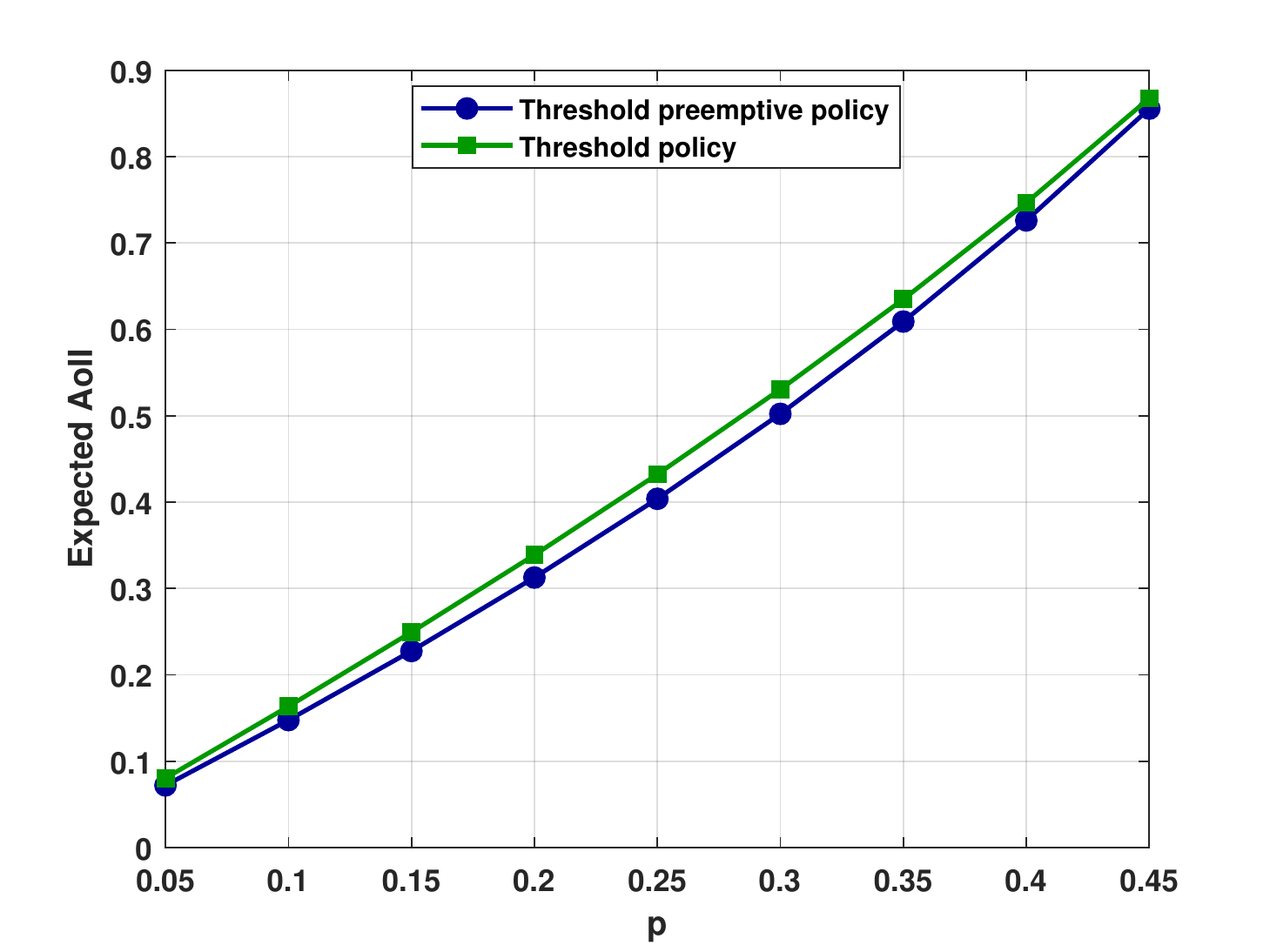}
\caption{When $a=3$ and $t_{max}=5$.}
\label{fig-PerformanceZipfp}
\end{subfigure}\hfill
\begin{subfigure}{3in}
\centering
\includegraphics[width=3in]{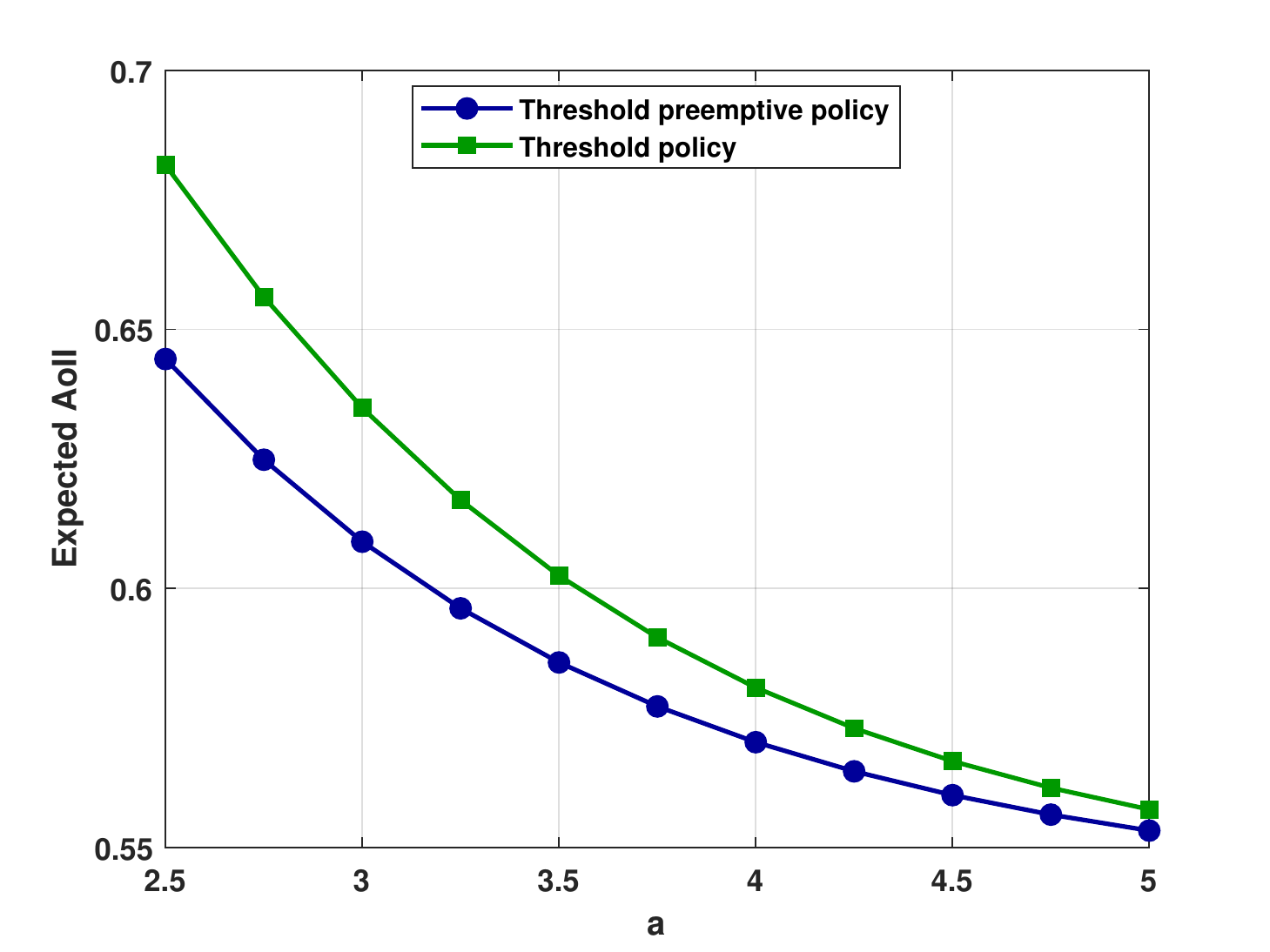}
\caption{When $p=0.35$ and $t_{max}=5$.}
\label{fig-PerformanceZipfa}
\end{subfigure}\hfill
\begin{subfigure}{3in}
\centering
\includegraphics[width=3in]{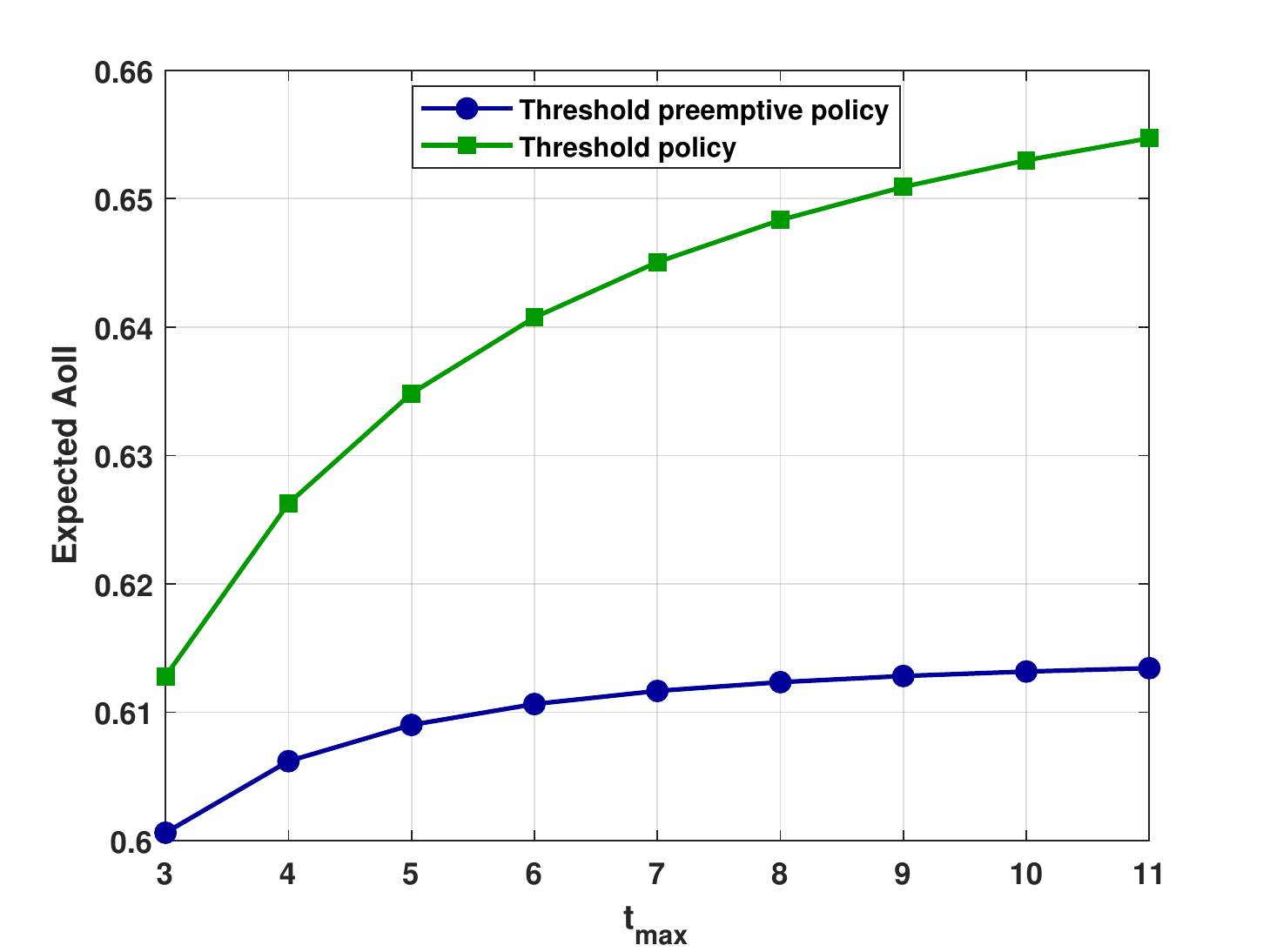}
\caption{When $p=0.35$ and $a=3$.}
\label{fig-PerformanceZipftmax}
\end{subfigure}\hfill
\caption{The performance comparison when the transmission delay follows the Zipf distribution. In this case, there are three system parameters: the Markovian source dynamics $p$, the constant $a$ in the Zipf distribution, and the upper bound on the transmission time $t_{max}$. We fix two of these parameters in the calculations, then vary the remaining one and plot the corresponding results.}
\label{fig-PerformanceZipf}
\end{figure}
Again, the performance gain from the preemption capability is not significant. The expected AoII achieved by the optimal policy grows with $p$ for the same reason as in the case of the Geometric distribution. As $a$ decreases and $t_{max}$ increases, the expected transmission time of an update increases, which leads to a decrease in the number of updates received by the receiver per unit of time. Therefore, the expected AoII increases.

\section{Conclusion}
In this paper, we optimize the performance of a transmitter-receiver pair in a system using the AoII with a generic time penalty function as the performance metric. In the system, the transmitter decides when to transmit status updates about a Markovian source to a distant receiver over a channel with a random delay to achieve the minimum expected AoII. The transmitter we consider can preempt the transmitting update to transmit a new update when the channel is busy, and the receiver will predict the state of the Markovian source based on the received update. First, we cast the optimization problem into an infinite horizon with average cost Markov decision process and provide the analytical expressions of the expected AoIIs achieved by two types of preemptive policy. Then, we prove the existence of the optimal policy and introduce the relative value iteration algorithm to find the optimal policy. To implement the relative value iteration algorithm, we truncate the Markov decision process so that its state space becomes finite. However, the optimal policy resulting from the relative value iteration algorithm is only an approximation. Therefore, we perform a theoretical analysis of the system when the delay distribution follows the Geometric and Zipf distributions, respectively. To this end, we introduce the policy iteration algorithm. Then, leveraging the policy improvement theorem, we theoretically find the corresponding optimal policies. For the system considered in this paper, it is always optimal to transmit new updates when the channel is idle. When the channel is transmitting an old update, whether the transmitter preempts is closely related to whether the transmitting update can bring new information to the receiver and whether the transmitting update carries correct information about the Markovian source. Finally, we present the numerical results on the validation of Condition~\ref{con}, the performance comparison between the optimal policy and the non-preemptive policy, and the effect of system parameters on the performance.

\bibliographystyle{IEEEtran}
\bibliography{mybib}

\newpage
\setcounter{page}{1}

\twocolumn[
\begin{@twocolumnfalse}
\begin{center}
\Huge Supplementary Material for the Paper "Preempting to Minimize Age of Incorrect Information under Transmission Delay"
\end{center}
\end{@twocolumnfalse}
\vspace{3em}]

\appendices
\section{Details of State Transition Probability}\label{app-STP}
For a clearer presentation, we write $P_{s,s'}(a)$ as $Pr[s'\mid s,a]$. Then, we distinguish between different states.
\begin{itemize}
\item $s = (0,0,-1)$. In this case, the channel is idle. We start with the case where the transmitter initiates a new transmission (i.e., $a=1$). We know that the update is delivered with probability $q_1$. In this case, $t'=0$ and $i'=-1$ by definition. Moreover, the receiver's estimate will not change. Hence, according to \eqref{eq-AoIIDynamic}, we have
\[
Pr[(0,0,-1)\mid (0,0,-1),a=1] = q_1(1-p).
\]
\[
Pr[(1,0,-1)\mid (0,0,-1),a=1] = q_1p.
\]
The update will still be in transmission with probability $1-q_1$. In this case, $t'=t+1$ as the transmission continues. $i'=0$ since the transmitted update is the same as the receiver's estimate. Meanwhile, the receiver's estimate will not change. Hence, according to \eqref{eq-AoIIDynamic}, we have
\[
Pr[(0,1,0)\mid (0,0,-1),a=1] = (1-q_1)(1-p).
\]
\[
Pr[(1,1,0)\mid (0,0,-1),a=1] = (1-q_1)p.
\]
Then, we consider the case where the transmitter chooses to stay idle (i.e., $a=0$). In this case, $t'=t$ and $i'=i$. Meanwhile, the receiver's estimate will remain the same as no update is delivered. Hence, according to \eqref{eq-AoIIDynamic}, we have
\[
Pr[(0,0,-1)\mid (0,0,-1),a=0] = 1-p.
\]
\[
Pr[(1,0,-1)\mid (0,0,-1),a=0] = p.
\]
\item $s = (0,t,0)$ where $t\geq 1$. In this case, the channel is busy. When the transmitter chooses to terminate the current transmission and initiate a new one (i.e., $a=1$), the update is delivered with probability $q_1$. In this case, $t'=0$ and $i'=-1$ by definition. Meanwhile, the receiver's estimate will not change. Hence, according to \eqref{eq-AoIIDynamic}, we have
\[
Pr[(0,0,-1)\mid (0,t,0),a=1] = q_1(1-p).
\]
\[
Pr[(1,0,-1)\mid (0,t,0),a=1] = q_1p.
\]
The update will still be in transmission with probability $1-q_1$. In this case, $t'=1$ because a new transmission starts. $i'=0$ because the transmitted update is the same as the receiver's estimate. Also, the receiver's estimate will not change. Hence, according to \eqref{eq-AoIIDynamic}, we have
\[
Pr[(0,1,0)\mid (0,t,0),a=1] = (1-q_1)(1-p).
\]
\[
Pr[(1,1,0)\mid (0,t,0),a=1] = (1-q_1)p.
\]
When the transmitter chooses $a=0$, the transmitted update will be delivered with probability $q_{t+1}$. In this case, $t'=0$ and $i'=-1$ by definition. Meanwhile, the receiver's estimate will not change as $i=0$ indicates that the newly arrived update brings no new information to the receiver. Hence, according to \eqref{eq-AoIIDynamic}, we have
\[
Pr[(0,0,-1)\mid (0,t,0),a=0] = q_{t+1}(1-p).
\]
\[
Pr[(1,0,-1)\mid (0,t,0),a=0] = q_{t+1}p.
\]
The transmitted update will still be in transmission with probability $1-q_{t+1}$. In this case, $t'=t+1$ as the transmission continues. $i'=i$, and the receiver's estimate will stay the same. Hence, according to \eqref{eq-AoIIDynamic}, we have
\[
Pr[(0,t+1,0)\mid (0,t,0),a=0] =(1-q_{t+1})(1-p).
\]
\[
Pr[(1,t+1,0)\mid (0,t,0),a=0] = (1-q_{t+1})p.
\]
\item $s = (0,t,1)$ where $t\geq 1$. The analysis is similar to the case of $s=(0,t,0)$ except that when the update is not preempted and is delivered, the receiver's estimate will flip. Hence, we present the results directly.
\[
Pr[(0,0,-1)\mid (0,t,1),a=1] = q_1(1-p).
\]
\[
Pr[(1,0,-1)\mid (0,t,1),a=1] = q_1p.
\]
\[
Pr[(0,1,0)\mid (0,t,1),a=1] = (1-q_1)(1-p).
\]
\[
Pr[(1,1,0)\mid (0,t,1),a=1] = (1-q_1)p.
\]
\[
Pr[(0,0,-1)\mid (0,t,1),a=0] = q_{t+1}p.
\]
\[
Pr[(1,0,-1)\mid (0,t,1),a=0] = q_{t+1}(1-p).
\]
\[
Pr[(0,t+1,1)\mid (0,t,1),a=0] = (1-q_{t+1})(1-p).
\]
\[
Pr[(1,t+1,1)\mid (0,t,1),a=0] = (1-q_{t+1})p.
\]
\item $s = (\Delta,t,i)$ where $\Delta>0$. In this case, the analysis is similar to the case of $s=(0,t,i)$ except for the following.
\begin{itemize}
\item $i'=1$ with probability $1-q_1$ when the transmitter chooses $a=1$.
\item When the receiver's estimate changes, $\Gamma=0$. Otherwise, $\Gamma=1$. Then, the dynamics of $\Delta'$ can be determined using \eqref{eq-AoIIDynamic}.
\end{itemize}
Hence, we omit the discussion and present the results directly.
\[
Pr[(0,0,-1)\mid (\Delta,0,-1),a=1] = q_1(1-p).
\]
\[
Pr[(\Delta+1,0,-1)\mid (\Delta,0,-1),a=1] = q_1p.
\]
\[
Pr[(0,1,1)\mid (\Delta,0,-1),a=1] = (1-q_1)p.
\]
\[
Pr[(\Delta+1,1,1)\mid (\Delta,0,-1),a=1] = (1-q_1)(1-p).
\]
\[
Pr[(\Delta+1,0,-1)\mid (\Delta,0,-1),a=0] = 1-p.
\]
\[
Pr[(0,0,-1)\mid (\Delta,0,-1),a=0] = p.
\]
For each $t\geq 1$,
\[
Pr[(0,0,-1)\mid (\Delta,t,0),a=1] = q_1(1-p).
\]
\[
Pr[(\Delta+1,0,-1)\mid (\Delta,t,0),a=1] = q_1p.
\]
\[
Pr[(0,1,1)\mid (\Delta,t,0),a=1] = (1-q_1)p.
\]
\[
Pr[(\Delta+1,1,1)\mid (\Delta,t,0),a=1] = (1-q_1)(1-p).
\]
\[
Pr[(0,0,-1)\mid (\Delta,t,0),a=0] = q_{t+1}p.
\]
\[
Pr[(\Delta+1,0,-1)\mid (\Delta,t,0),a=0] = q_{t+1}(1-p).
\]
\[
Pr[(0,t+1,0)\mid (\Delta,t,0),a=0] = (1-q_{t+1})p.
\]
\[
Pr[(\Delta+1,t+1,0)\mid (\Delta,t,0),a=0] = (1-q_{t+1})(1-p).
\]
\[
Pr[(0,0,-1)\mid (\Delta,t,1),a=1] = q_1(1-p).
\]
\[
Pr[(\Delta+1,0,-1)\mid (\Delta,t,1),a=1] = q_1p.
\]
\[
Pr[(0,1,1)\mid (\Delta,t,1),a=1] = (1-q_1)p.
\]
\[
Pr[(\Delta+1,1,1)\mid (\Delta,t,1),a=1] = (1-q_1)(1-p).
\]
\[
Pr[(0,0,-1)\mid (\Delta,t,1),a=0] = q_{t+1}(1-p).
\]
\[
Pr[(\Delta+1,0,-1)\mid (\Delta,t,1),a=0] = q_{t+1}p.
\]
\[
Pr[(0,t+1,1)\mid (\Delta,t,1),a=0] = (1-q_{t+1})p.
\]
\[
Pr[(\Delta+1,t+1,1)\mid (\Delta,t,1),a=0] = (1-q_{t+1})(1-p).
\]
\end{itemize}
Combining the cases together, we fully characterized the state transition probability $P_{s,s'}(a)$.

\section{Proof of Lemma~\ref{them-StationarySP}}\label{pf-StationarySP}
Combining with the system dynamics, the steady state probabilities satisfy the following balance equations.
\begin{equation}\label{eq-preemptiveperformance1}
\pi_{-1}(0) = q_1(1-p)\sum_{i=0}^{\infty}\pi(i).
\end{equation}
\[
\pi_{-1}(\Delta) = q_1p\pi(\Delta-1)\quad\Delta\geq1.
\]
\[
\pi_{0}(0) = (1-q_1)(1-p)\pi(0).
\]
\[
\pi_{0}(1) = (1-q_1)p\pi(0).
\]
\[
\pi_0(\Delta) = 0\quad\Delta\geq2.
\]
\[
\pi_{1}(0) = (1-q_1)p\sum_{i=1}^{\infty}\pi(i).
\]
\[
\pi_{1}(1) = 0.
\]
\[
\pi_{1}(\Delta) = (1-q_1)(1-p)\pi(\Delta-1)\quad\Delta\geq2.
\]
\begin{equation}\label{eq-preemptiveperformance4}
\sum_{i=0}^{\infty}\pi(i) = 1.
\end{equation}
Combining \eqref{eq-preemptiveperformance1} and \eqref{eq-preemptiveperformance4} yields
\[
\pi_{-1}(0) = q_1(1-p).
\]
According to the definition of $\pi(0)$, we have
\[
\pi(0) = q_1(1-p) + (1-q_1)(1-p)\pi(0) + (1-q_1)p(1-\pi(0)).
\]
Then, we obtain
\[
\pi(0)=\frac{p + q_1 - 2q_1p}{1-(1-q_1)(1-2p)}.
\]
Likewise, we can obtain
\[
\begin{split}
\pi(1) = & q_1p\pi(0) + (1-q_1)p\pi(0) = p\pi(0)\\
= & \frac{p^2 + q_1p - 2q_1p^2}{1-(1-q_1)(1-2p)}.
\end{split}
\]
\begin{equation}\label{eq-preemptiveperformance5}
\begin{split}
\pi(\Delta) = & (q_1p + (1-q_1)(1-p))\pi(\Delta-1)\\
= & (1-q_1-p + 2q_1p)\pi(\Delta-1)\quad \Delta\geq2.
\end{split}
\end{equation}
After some algebraic manipulation, for each $\Delta\geq1$, we have
\[
\begin{split}
\pi(\Delta) = & (1-q_1-p + 2q_1p)^{\Delta-1}\pi(1)\\
= & \frac{(1-q_1-p + 2q_1p)^{\Delta-1}(p^2 + q_1p - 2q_1p^2)}{1-(1-q_1)(1-2p)}.
\end{split}
\]

\section{Proof of Corollary~\ref{cor-SpecialPerformanceSP}}\label{pf-SpecialPerformanceSP}
We derive the closed-form expression based on Lemma~\ref{them-StationarySP} and its proof. We define $\Pi \triangleq \sum_{\Delta=2}^{\infty}\pi(\Delta)$. Then, we sum \eqref{eq-preemptiveperformance5} from $2$ to $\infty$ and apply the definition of $\Pi$, which yield
\[
\Pi = (1-q_1-p+2q_1p)(\Pi+\pi(1)).
\]
After some algebraic manipulations, we have
\[
\Pi = \frac{1-q_1-p+2q_1p}{q_1+p - 2q_1p}\pi(1).
\]
We also define $\Sigma \triangleq \sum_{\Delta=2}^{\infty}\Delta\pi(\Delta)$. Then, the expected AoII achieved by the strong preemptive policy is
\[
\bar{\Delta}_{sp}= \alpha\left(\pi(1) + \Sigma\right)+\beta.
\]
To obtain $\Sigma$, we multiply both size of \eqref{eq-preemptiveperformance5} by $\Delta-1$.
\[
(\Delta-1)\pi(\Delta) = (1-q_1-p + 2q_1p)(\Delta-1)\pi(\Delta-1)\quad \Delta\geq2.
\]
Then, we sum the above equation from $2$ to $\infty$ and apply the definitions of $\Pi$ and $\Sigma$, which yield
\[
\Sigma - \Pi = (1-q_1-p + 2q_1p)(\Sigma + \pi(1)).
\]
Then, we obtain
\[
\Sigma = \frac{(1-q_1-p + 2q_1p)\pi(1) + \Pi}{q_1+p - 2q_1p}.
\]
Plugging in the expressions of $\pi(1)$ and $\Pi$, we obtain
\[
\Sigma = \ddfrac{p-p(q_1+p-2q_1p)^2}{(p+q_1-2q_1p)(q_1+2p-2q_1p)}.
\]
Combining together, we have
\[
\bar{\Delta}_{sp} = \ddfrac{\alpha p}{(p+q_1-2q_1p)(q_1+2p-2q_1p)}+\beta.
\]

\section{Proof of Lemma~\ref{lem-StationaryWP}}\label{pf-StationaryWP}
We combine \eqref{eq-uniformperformance2} and \eqref{eq-uniformperformance3}, which yields
\[
\Pi(t) = \prod_{l=1}^t(1-q_l)(1-p)\Pi \triangleq \prod_{l=1}^t\mathcal{P}_l\Pi\quad 1\leq t\leq t_{max}-1.
\]
Then, plugging in the result into \eqref{eq-uniformperformance1} gives us
\[
\Pi - p\pi_0 = q_1p\Pi + \left(\sum_{t=1}^{t_{max}-1}q_{t+1}p \prod_{l=1}^t\mathcal{P}_l\right)\Pi.
\]
Then, we can obtain
\[
\pi_0 = \ddfrac{1-q_1p - p\left[\sum_{t=1}^{t_{max}-1}q_{t+1} \left(\prod_{l=1}^t\mathcal{P}_l\right)\right]}{p}\Pi.
\]
Since we have expressed $\pi(0)$ and $\Pi(t)$ for $1\leq t\leq t_{max}-1$ using $\Pi$, combining with \eqref{eq-uniformperformance7} yields \eqref{eq-EquivalentEq6}.
\begin{figure*}[!t]
\normalsize
\begin{equation}\label{eq-EquivalentEq6}
\Pi = \ddfrac{1}{\frac{1}{p}-q_1 - \left[\sum_{t=1}^{t_{max}-1}q_{t+1}\left(\prod_{l=1}^t\mathcal{P}_l\right)\right] + 1 + \sum_{t=1}^{t_{max}-1}\left(\prod_{l=1}^t\mathcal{P}_l\right)}.
\end{equation}
\hrulefill
\vspace*{4pt}
\end{figure*}

\section{Proof of Theorem~\ref{thm-SpecialPerformanceWP}}\label{pf-SpecialPerformanceWP}
We first multiply both sides of \eqref{eq-uniformperformance6} by $\Delta-1$.
\begin{multline*}
(\Delta-1)\pi_\Delta(t) = (1-q_t)(1-p)(\Delta-1)\pi_{\Delta-1}(t-1)\\
2\leq t\leq t_{max}-1\ and\ \Delta\geq2.
\end{multline*}
Then, we sum the above equation over $\Delta$ from $2$ to $\infty$.
\begin{multline*}
\sum_{\Delta=2}^{\infty}(\Delta-1)\pi_\Delta(t)= (1-q_t)(1-p)\sum_{\Delta=1}^{\infty}\Delta\pi_{\Delta}(t-1)\\
2\leq t\leq t_{max}-1.
\end{multline*}
Plugging in the definitions, we obtain
\begin{multline}\label{eq-uniformperformance8}
\Sigma(t)-\Pi(t)= (1-q_t)(1-p)\Sigma(t-1)=\mathcal{P}_t\Sigma(t-1)\\
2\leq t\leq t_{max}-1.
\end{multline}
By applying the same steps to \eqref{eq-uniformperformance5}, we can obtain the following.
\begin{equation}\label{eq-uniformperformance9}
\Sigma(1) - \Pi(1) = (1-q_1)(1-p)\Sigma = \mathcal{P}_1\Sigma.
\end{equation}
Combining \eqref{eq-uniformperformance8} and \eqref{eq-uniformperformance9} gives us the following.
\begin{multline}\label{eq-uniformperformance10}
\Sigma(t) = \left(\prod_{l=1}^{t}\mathcal{P}_l\right)\Sigma + \sum_{i=1}^{t}\left[\left(\prod_{j=i+1}^t
\mathcal{P}_j\right)\Pi(i)\right]\\
1\leq t\leq t_{max}-1.
\end{multline}
Then, we apply again the same steps to \eqref{eq-uniformperformance4}.
\begin{equation}\label{eq-uniformperformance11}
\Sigma - \Pi = p_1p\Sigma + \sum_{t=1}^{t_{max}-1}p_{t+1}p\Sigma(t).
\end{equation}
When combining \eqref{eq-uniformperformance10} and \eqref{eq-uniformperformance11}, we can obtain
\begin{multline*}
\left\lbrace1-p_1p-\sum_{t=1}^{t_{max}-1}\left[p_{t+1}p\left(\prod_{l=1}^t\mathcal{P}_l\right)\right]\right\rbrace\Sigma =\\
 \Pi + \sum_{t=1}^{t_{max}-1}\left\lbrace p_{t+1}p\left[\sum_{i=1}^t\left(\prod_{j=i+1}^t\mathcal{P}_j\right)\Pi(i)\right]\right\rbrace.
\end{multline*}
Rearranging the terms yields
\begin{equation}\label{eq-uniformperformance12}
\Sigma = \ddfrac{\Pi + \sum_{t=1}^{t_{max}-1}\left\{p_{t+1}p\left[\sum_{i=1}^t\left(\prod_{j=i+1}^t\mathcal{P}_j\right)\Pi(i)\right]\right\}}{1-p_1p-\sum_{t=1}^{t_{max}-1}\left[p_{t+1}p\left(\prod_{l=1}^t\mathcal{P}_l\right)\right]}.
\end{equation}
Finally, plugging \eqref{eq-uniformperformance10} and \eqref{eq-uniformperformance12} into \eqref{eq-AoIIwp} gives us the closed-form expression of $\bar{\Delta}_{wp}$.

\section{Proof of Lemma~\ref{lem-Monotone}}\label{pf-Monotone}
Given that the value function can be computed iteratively, we use mathematical induction to prove the desired result. First, the base case $\nu=0$ is true by initialization. Then, we assume that the monotonicity holds at iteration $\nu$ and check whether the monotonicity still holds at iteration $\nu+1$. To this end, we first revisit how the estimated value function is updated by incorporating the structural properties of the state transition probability. From Appendix~\ref{app-STP}, we know that, within a single transition, $\Delta$ either increases by one or decreases to zero. More precisely,
\[
Pr[(\Delta',t',i')\mid (\Delta,t,i),a] = 0\quad\Delta'\notin\{0,\Delta+1\}.
\]
Applying the structural property to \eqref{eq-ValueIterationGammaDiscount} yields \eqref{eq-EquivalentEq3}.
\begin{figure*}[!t]
\normalsize
\begin{multline}\label{eq-EquivalentEq3}
V_{\gamma,\nu+1}(s) = \min_{a\in\mathcal{A}}\bigg\{C(s)+\gamma\sum_{t',i'}\bigg(Pr[(\Delta+1,t',i')\mid(\Delta,t,i),a]V_{\gamma,\nu}(\Delta+1,t',i') + \\
\hspace{10em} Pr[(0,t',i')\mid(\Delta,t,i),a]V_{\gamma,\nu}(0,t',i')\bigg)\bigg\}\quad s\in\mathcal{S}.
\end{multline}
\hrulefill
\vspace*{4pt}
\end{figure*}
Moreover, the state transition probability is independent of $\Delta$ when $\Delta>0$. Specifically, for any $\Delta_1>0$, $\Delta_2>0$, $t$, and $i$,
\[
Pr[(0,t'i')\mid (\Delta_1,t,i),a] = Pr[(0,t'i')\mid (\Delta_2,t,i),a].
\]
\begin{multline*}
Pr[(\Delta_1+1,t',i')\mid (\Delta_1,t,i),a] = \\
Pr[(\Delta_2+1,t',i')\mid (\Delta_2,t,i),a].
\end{multline*}
Let $V_{\gamma,\nu+1}^a(s)$ be the value function resulting from the adoption of action $a$, $s_1=(\Delta_1,t,i)$, and $s_2=(\Delta_2,t,i)$ where $\Delta_1\geq\Delta_2>0$. We notice that the immediate cost $C(s)$ depends only on $\Delta$ and is non-decreasing in $\Delta$. Combined with the assumption on the estimated value function at iteration $\nu$, we can conclude that
\[
V_{\gamma,\nu+1}^a(s_1) \geq V_{\gamma,\nu+1}^a(s_2)\quad a\in\{0,1\}.
\]
Since $V_{\gamma,\nu+1}(s) = \min_{a\in\{0,1\}}\{V_{\gamma,\nu+1}^a(s)\}$, we can conclude that the monotonicity still holds at iteration $\nu+1$. Then, by mathematical induction, the lemma is true.

\section{Proof of Theorem~\ref{thm-policyimprovement}}\label{pf-policyimprovement}
The proof is based on~\cite[pp.42-43]{b18}. We consider a generic MDP $\mathcal{M}^G$. We first clarify the notations we will use in the proof.
\begin{itemize}
\item The state space and the state is denoted by $\mathcal{S}^G$ and $s$, respectively. The action suggested by policy $A$ at state $s$ is denoted by $A(s)$.
\item The probability that operating policy $A$ at state $s$ leads to state $s'$ is denoted by $P_{s,s'}^A$. Likewise, the probability that action $a$ at state $s$ leads to state $s'$ is denoted by $P_{s,s'}(a)$.
\item The immediate cost for operating policy $A$ at state $s$ is denoted by $C(s,A)$. Similarly, the immediate cost for operating action $a$ at state $s$ is denoted by $C(s,a)$
\item The value function of state $s$ resulting from the adoption of policy $A$ is denoted by $V^A(s)$. The expected cost achieved by policy $A$ is denoted by $\theta^A$.
\end{itemize}
With the notations in mind, we prove the optimality by contradiction. We assume that the policy improvement function has converged to policy $A$ and there exists a policy $B$ such that $\theta^A>\theta^B$.

We recall that the policy improvement function procedures a new policy $\psi'$ based on the old policy $\psi$ using the following equation.
\[
\psi'(s) = \argmin_{a\in\mathcal{A}}\left\{C(s,a) + \sum_{s'\in\mathcal{S}^G}P_{s,s'}(a)V^{\psi}(s')\right\}.
\]
Since the policy improvement function has converged to policy $A$, we have the following inequality holds for any policy $B$.
\[
C(s,A) + \sum_{s'\in\mathcal{S}^G}P_{s,s'}^AV^A(s') \leq C(s,B) + \sum_{s'\in\mathcal{S}^G}P_{s,s'}^BV^A(s').
\]
We define
\[
\delta(s)\triangleq C(s,A) - C(s,B) + \sum_{s'\in\mathcal{S}^G}(P_{s,s'}^A - P_{s,s'}^B)V^A(s')\leq0.
\]
Meanwhile, $V^A(s)$ and $V^B(s)$ satisfy there own Bellman equations. More precisely,
\[
V^A(s) + \theta^A = C(s,A) + \sum_{s'\in\mathcal{S}^G}P_{s,s'}^AV^A(s').
\]
\[
V^B(s) + \theta^B = C(s,B) + \sum_{s'\in\mathcal{S}^G}P_{s,s'}^BV^B(s').
\]
Subtracting the above two equations and bringing in $\delta(s)$ yield
\begin{multline*}
V^A(s) -V^B(s) + \theta^A -\theta^B = \\
\delta(s) + \sum_{s'\in\mathcal{S}^G}P_{s,s'}^B(V^A(s')-V^B(s')).
\end{multline*}
Let $V^{\delta}(s) \triangleq V^A(s) -V^B(s)$ and $\theta^{\delta}\triangleq \theta^A -\theta^B$. Plugging in the definitions yields
\[
V^{\delta}(s) + \theta^{\delta} = \delta(s) + \sum_{s'\in\mathcal{S}^G}P_{s,s'}^BV^{\delta}(s').
\]
As is mentioned in Section~\ref{sec-Performance}, each policy induces a DTMC and the expected cost $\theta^{\delta} = \sum_{s\in\mathcal{S}^G}\delta(s)\pi^B(s)$ where $\pi^B(s)$ is the stationary distribution of the DTMC induced by policy $B$. Since the stationary distribution is non-negative and $\delta(s)\leq0$ for all $s\in\mathcal{M}^G$, we can conclude that $\theta^{\delta}\leq0$. In other words, $\theta^A\leq\theta^B$, which contradicts the assumption that $\theta^A>\theta^B$. Therefore, the contradiction occurs, and the converged policy $A$ is optimal.

\section{Proof of Theorem~\ref{thm-optimalgeometric}}\label{pf-optimalgeometric}
We first calculate the value function $V^{\psi}(s)$ resulting from adopting the strong preemptive policy $\psi$. Combining the strong preemptive policy definition and \eqref{eq-PolicyIteration1}, we know that the value function satisfies the following system of linear equations.
\begin{align*}
V^{\psi}(0,0,-1) = f(0)- \theta^{\psi} + & (1-p_s)(1-p)V^{\psi}(0,1,0) + \\
& (1-p_s)pV^{\psi}(1,1,0) + \\
& p_s(1-p)V^{\psi}(0,0,-1) + \\
& p_spV^{\psi}(1,0,-1).
\end{align*}
\begin{align*}
V^{\psi}(0,t,0) = f(0)- \theta^{\psi} + & (1-p_s)(1-p)V^{\psi}(0,1,0) + \\
& (1-p_s)pV^{\psi}(1,1,0) + \\
& p_s(1-p)V^{\psi}(0,0,-1)+\\
& p_spV^{\psi}(1,0,-1)\quad t\geq1.
\end{align*}
\begin{align*}
V^{\psi}(0,t,1) = f(0)- \theta^{\psi} + & (1-p_s)(1-p)V^{\psi}(0,1,0) + \\
& (1-p_s)pV^{\psi}(1,1,0) + \\
& p_s(1-p)V^{\psi}(0,0,-1)+ \\
& p_spV^{\psi}(1,0,-1)\quad t\geq1.
\end{align*}
For each $\Delta\geq1$,
\begin{align*}
V^{\psi}(\Delta,0,-1) = & f(\Delta) - \theta^{\psi} + \\
& (1-p_s)(1-p)V^{\psi}(\Delta+1,1,1) + \\
& (1-p_s)pV^{\psi}(0,1,1) + \\
& p_s(1-p)V^{\psi}(0,0,-1)+ \\
& p_spV^{\psi}(\Delta+1,0,-1).
\end{align*}
\begin{align*}
V^{\psi}(\Delta,t,0) = & f(\Delta) - \theta^{\psi} +\\
& (1-p_s)(1-p)V^{\psi}(\Delta+1,1,1) + \\
& (1-p_s)pV^{\psi}(0,1,1) + \\
& p_s(1-p)V^{\psi}(0,0,-1)+\\
& p_spV^{\psi}(\Delta+1,0,-1)\quad t\geq1.
\end{align*}
\begin{align*}
V^{\psi}(\Delta,t,1) = & f(\Delta) - \theta^{\psi} + \\
& (1-p_s)(1-p)V^{\psi}(\Delta+1,1,1) + \\
& (1-p_s)pV^{\psi}(0,1,1) + \\
& p_s(1-p)V^{\psi}(0,0,-1)+\\
& p_spV^{\psi}(\Delta+1,0,-1)\quad t\geq1.
\end{align*}
We notice that the size of the system of linear equations is infinite, so it is difficult to obtain the solution by solving directly. However, some structural properties of $V^{\psi}(s)$ are sufficient for us to complete the proof. These structural properties are summarized in the following lemma.
\begin{lemma}\label{lem-geometricvfproperty}
$V^{\psi}(s)$ possesses the following structural properties.
\begin{enumerate}
\item $V^{\psi}(\Delta,0,-1) = V^{\psi}(\Delta,t,0) = V^{\psi}(\Delta,t,1) \triangleq V^{\psi}(\Delta)$ for $\Delta\geq0$ and $t\geq 1$.
\item $V^{\psi}(\Delta)$ is non-decreasing in $\Delta$.
\end{enumerate}
\end{lemma}
\begin{proof}
The first property can be verified easily by comparing the linear equations they satisfy. Hence, we will focus on the second property. The system of linear equations can be solved iteratively~\cite{b15}. More precisely,
\[
V_{\nu+1}^{\psi}(s) = f(\Delta) -\theta^\psi+ \sum_{s'\in\mathcal{S}}P_{s,s'}^{\psi}V^{\psi}_{\nu}(s')\quad s\in\mathcal{S},
\]
where $V_{\nu}^{\psi}(s)$ is the estimated value function at iteration $\nu$ and $P_{s,s'}^{\psi}$ is the probability that operating policy $\psi$ at state $s$ leads to state $s'$. We know that $\lim_{\nu\rightarrow\infty}V^{\psi}_{\nu}(s) = V^{\psi}(s)$. Then, leveraging the iterative nature, we can use mathematical induction to prove the desired results. We initialize $V^{\psi}_0(s)=0$ for $s\in\mathcal{S}$. Then, the base case $\nu=0$ is true by initialization. We assume the monotonicity is true at iteration $\nu$. Then, we check if it holds at iteration $\nu+1$. Using the first property, we have the following holds for state $s$ with $\Delta\geq1$.
\begin{multline*}
V^{\psi}_{\nu+1}(\Delta+1)-V^{\psi}_{\nu+1}(\Delta) = f(\Delta+1)-f(\Delta) + \\
(1-p_s)(1-p)[V^{\psi}_{\nu}(\Delta+2)-V^{\psi}_{\nu}(\Delta+1)] + \\
p_sp[V^{\psi}_{\nu}(\Delta+2)-V^{\psi}_{\nu}(\Delta+1)].
\end{multline*}
Applying the assumption for iteration $\nu$ and the monotonicity of the time penalty function, we can easily conclude that $V^{\psi}_{\nu+1}(\Delta+1)\geq V^{\psi}_{\nu+1}(\Delta)$ when $\Delta\geq1$. Then, we consider the case of $\Delta=0$.
\[
\begin{split}
V&^{\psi}_{\nu+1}(1)-V^{\psi}_{\nu+1}(0) = f(1)-f(0) + \\
& (1-p_s)(1-p)[V_{\nu}^{\psi}(2)-V_{\nu}^{\psi}(0)] + \\
& (1-p_s)p[V_{\nu}^{\psi}(0)-V_{\nu}^{\psi}(1)] +p_sp[V_{\nu}^{\psi}(2)-V_{\nu}^{\psi}(1)]\\
\geq & (1-p_s)(1-p)[V_{\nu}^{\psi}(1)-V_{\nu}^{\psi}(0)] + p_sp[V_{\nu}^{\psi}(2)-V_{\nu}^{\psi}(1)] + \\
& (1-p_s)p[V_{\nu}^{\psi}(0)-V_{\nu}^{\psi}(1)]\\
= & (1-p_s)(1-2p)[V_{\nu}^{\psi}(1)-V_{\nu}^{\psi}(0)] + p_sp[V_{\nu}^{\psi}(2)-V_{\nu}^{\psi}(1)].
\end{split}
\]
We recall that $0<p<\frac{1}{2}$. Hence, $V^{\psi}_{\nu+1}(1)\geq V^{\psi}_{\nu+1}(0)$. Combining together, the property holds at iteration $\nu+1$. Then, by mathematical induction, we can conclude that the second property is true.
\end{proof}
Equipped with Lemma~\ref{lem-geometricvfproperty}, we can proceed to obtain the new policy induced by the value function $V^{\psi}(s)$. To this end, we define $V^{\psi,a}(s)$ as the expected cost resulting from taking action $a$ at state $s$, which can be calculated using the following equation.
\[
\begin{split}
V^{\psi,a}(s) = & f(\Delta) - \theta^{\psi} + \sum_{s'\in\mathcal{S}}P_{s,s'}(a)V^{\psi}(s').
\end{split}
\]
Consequently, to determine the new action, we only need to determine the sign of $\delta V^{\psi}(s)\triangleq V^{\psi,0}(s) - V^{\psi,1}(s)$. When $\delta V^{\psi}(s)<0$, the action suggested by the new policy is $a=0$. Otherwise, $a=1$ is suggested. Without loss of generality, we let $V^{\psi}(0)=0$. Then, we distinguish between different states.
\begin{itemize}
\item $s = (0,0,-1)$.
\begin{align*}
\delta V^{\psi}(0,0,-1) = & pV^{\psi}(1) - (1-p_s)pV^{\psi}(1) -\\
& p_spV^{\psi}(1)\\
= & 0.
\end{align*}
\item $s = (\Delta,0,-1)$ where $\Delta\geq1$.
\begin{align*}
\delta V^{\psi}(\Delta,0,-1) = & (1-p)V^{\psi}(\Delta+1) - \\
& (1-p_s)(1-p)V^{\psi}(\Delta+1) - \\
& p_spV^{\psi}(\Delta+1) \\
= & p_s(1-2p)V^{\psi}(\Delta+1)\\
\geq & 0.
\end{align*}
\item $s = (0,t,0)$ where $t\geq 1$.
\begin{align*}
\delta V^{\psi}(0,t,0) = & (1-p_s)pV^{\psi}(1) + p_spV^{\psi}(1) - \\
& (1-p_s)pV^{\psi}(1)- p_spV^{\psi}(1)\\
= & 0.
\end{align*}
\item $s = (0,t,1)$ where $t\geq 1$.
\begin{align*}
\delta V^{\psi}(0,t,1)  = & (1-p_s)pV^{\psi}(1) + p_s(1-p)V^{\psi}(1) - \\
& (1-p_s)pV^{\psi}(1) - p_spV^{\psi}(1)\\
= & p_s(1-2p)V^{\psi}(1)\\
\geq & 0.
\end{align*}
\item $s = (\Delta,t,0)$ where $\Delta\geq1$ and $t\geq 1$.
\begin{align*}
\delta V^{\psi}(\Delta,t,0) = & (1-p_s)(1-p)V^{\psi}(\Delta+1) +\\
& p_s(1-p)V^{\psi}(\Delta+1) - \\
& (1-p_s)(1-p)V^{\psi}(\Delta+1) - \\
& p_spV^{\psi}(\Delta+1)\\
= & p_s(1-2p)V^{\psi}(\Delta+1)\\
\geq & 0.
\end{align*}
\item $s = (\Delta,t,1)$ where $\Delta\geq1$ and $t\geq 1$.
\begin{align*}
\delta V^{\psi}(\Delta,t,1) = & (1-p_s)(1-p)V^{\psi}(\Delta+1)+\\
& p_spV^{\psi}(\Delta+1)-\\
& (1-p_s)(1-p)V^{\psi}(\Delta+1)  - \\
& p_spV^{\psi}(\Delta+1)\\
= & 0.
\end{align*}
\end{itemize} 
Combing together, we know that $\delta V^{\psi}(s)\geq0$ for all $s\in\mathcal{S}$, meaning that the new policy always suggests the cation $a=1$. Hence, the new policy is still the strong preemptive policy. Then, by Theorem~\ref{thm-policyimprovement}, we can conclude that the strong preemptive policy is optimal.

\section{Proof of Theorem~\ref{thm-optimalzipf}}\label{proof-optimalzipf}
We follow the same methodology presented in the proof of Theorem~\ref{thm-optimalgeometric}. First, we calculate the value function $V^{\psi}(s)$ resulting from the adoption of the threshold preemptive policy $\psi$. The value function satisfies the following system of linear equations.
\begin{align*}
V^{\psi}(0,0,-1) = f(0)- \theta^{\psi} + & (1-q_1)(1-p)V^{\psi}(0,1,0) + \\
& (1-q_1)pV^{\psi}(1,1,0) + \\
& q_1(1-p)V^{\psi}(0,0,-1)+\\
& q_1pV^{\psi}(1,0,-1).
\end{align*}
\begin{align*}
V^{\psi}(0,t,0) = & f(0)- \theta^{\psi} +\\
& (1-q_1)(1-p)V^{\psi}(0,1,0) + \\
& (1-q_1)pV^{\psi}(1,1,0) + \\
& q_1(1-p)V^{\psi}(0,0,-1)+\\
& q_1pV^{\psi}(1,0,-1)\quad 1\leq t\leq t_{max}-1.
\end{align*}
\begin{align*}
V^{\psi}(0,t,1) = & f(0)- \theta^{\psi} + \\
& (1-q_1)(1-p)V^{\psi}(0,1,0) + \\
& (1-q_1)pV^{\psi}(1,1,0) + \\
& q_1(1-p)V^{\psi}(0,0,-1)+\\
& q_1pV^{\psi}(1,0,-1)\quad 1\leq t\leq t_{max}-1.
\end{align*}
For each $\Delta\geq1$,
\begin{align*}
V^{\psi}(\Delta,0,-1) = & f(\Delta) - \theta^{\psi} + \\
& (1-q_1)(1-p)V^{\psi}(\Delta+1,1,1) + \\
& (1-q_1)pV^{\psi}(0,1,1) + \\
& q_1(1-p)V^{\psi}(0,0,-1)+\\
& q_1pV^{\psi}(\Delta+1,0,-1).
\end{align*}
\begin{align*}
V^{\psi}(\Delta,t,0) = & f(\Delta) - \theta^{\psi} + \\
& (1-q_1)(1-p)V^{\psi}(\Delta+1,1,1) + \\
& (1-q_1)pV^{\psi}(0,1,1) + \\
& q_1(1-p)V^{\psi}(0,0,-1)+\\
& q_1pV^{\psi}(\Delta+1,0,-1)\quad 1\leq t\leq t_{max}-1.
\end{align*}
\begin{align*}
V^{\psi}(\Delta,t,1) = & f(\Delta) - \theta^{\psi} + \\
& (1-q_1)(1-p)V^{\psi}(\Delta+1,1,1) + \\
& (1-q_1)pV^{\psi}(0,1,1) + \\
& q_1(1-p)V^{\psi}(0,0,-1)+\\
& q_1pV^{\psi}(\Delta+1,0,-1)\quad 1\leq t\leq t_{max}-2.
\end{align*}
\begin{multline}\label{eq-zipfoptimalproof4}
V^{\psi}(\Delta,t_{max}-1,1) = f(\Delta) - \theta^{\psi} + \\
(1-p)V^{\psi}(0,0,-1)+pV^{\psi}(\Delta+1,0,-1).
\end{multline}
Instead of solving the above system of linear equations, some structural properties of the solution will be sufficient for the following analysis.
\begin{lemma}\label{lem-zipfvfproperty}
$V^{\psi}(s)$ possesses the following properties.
\begin{enumerate}
\item $V^{\psi}(\Delta,0,-1) = V^{\psi}(\Delta,t,0)\triangleq V^{\psi}(\Delta)$ for $\Delta\geq0$ and $1\leq t\leq t_{max}-1$.
\item $V^{\psi}(0,t,1) = V^{\psi}(0)$ for $1\leq t\leq t_{max}-1$. $V^{\psi}(\Delta,t,1) = V^{\psi}(\Delta)$ for $\Delta>0$ and $1\leq t\leq t_{max}-2$. 
\item $V^{\psi}(\Delta)$ is non-decreasing in $\Delta$.
\item $V^{\psi}(1) - V^{\psi}(0) = \frac{\theta^{\psi}-f(0)}{p}$ and $V^{\psi}(\Delta+1) - V^{\psi}(\Delta) \triangleq \sigma$ is independent of $\Delta\geq1$, where
\[
\sigma = \frac{\alpha}{q_1+p-2q_1p}.
\]
\end{enumerate}
\end{lemma}
\begin{proof}
The first two properties are obvious, as we can verify directly by comparing the corresponding linear equations. For the third property, the proof is based on the mathematical induction as presented in the proof of Lemma \ref{lem-geometricvfproperty}. Hence, we omit the proof here for the sake of space. In the following, we focus on the last property. Applying the first two properties to the system of linear equations yields
\begin{equation}\label{eq-zipfoptimalproof1}
V^{\psi}(0) = f(0)-\theta^{\psi} + (1-p)V^{\psi}(0) + pV^{\psi}(1).
\end{equation}
\begin{align*}
V^{\psi}(\Delta) = & f(\Delta) - \theta^{\psi} + (1-q_1)(1-p)V^{\psi}(\Delta+1) + \\
& (1-q_1)pV^{\psi}(0) +  q_1(1-p)V^{\psi}(0)+ \\
& q_1pV^{\psi}(\Delta+1)\quad\Delta\geq1.
\end{align*}
From \eqref{eq-zipfoptimalproof1}, we can easily conclude that
\[
V^{\psi}(1) - V^{\psi}(0) = \frac{\theta^{\psi}-f(0)}{p}.
\]
In the following, we prove that $V^{\psi}(\Delta+1) - V^{\psi}(\Delta)$ is independent of $\Delta\geq1$. We recall that $V^{\psi}(\Delta)$ can be calculated using the iterative method. Let $V^{\psi}_{\nu}(\Delta)$ be the estimated value function at iteration $\nu$, which is updated in the following way.
\[
V_{\nu+1}^{\psi}(0) = f(0)-\theta^{\psi} + (1-p)V_{\nu}^{\psi}(0) + pV_{\nu}^{\psi}(1).
\]
\begin{align*}
V_{\nu+1}^{\psi}(\Delta) = & f(\Delta) - \theta^{\psi} + (1-q_1)(1-p)V_{\nu}^{\psi}(\Delta+1) + \\
& (1-q_1)pV_{\nu}^{\psi}(0) + q_1(1-p)V_{\nu}^{\psi}(0)+\\
& q_1pV_{\nu}^{\psi}(\Delta+1)\quad\Delta\geq1.
\end{align*}
Then, we know that $\lim_{\nu\rightarrow\infty}V_{\nu}^{\psi}(\Delta) = V^{\psi}(\Delta)$. Consequently, we can use mathematical induction to prove the desired results. To this end, we initialize $V_{0}^{\psi}(\Delta)=0$ for $\Delta\geq0$. Then, the base case $\nu=0$ is true by initialization. We assume the property holds at iteration $\nu$ and examine whether it still holds at iteration $\nu+1$. We recall that $f(\Delta) = \alpha\Delta+\beta$. Hence, we have
\begin{multline*}
V^{\psi}_{\nu+1}(\Delta+1) - V^{\psi}_{\nu+1}(\Delta) = \\
\alpha + (1-q_1)(1-p)[V_{\nu}^{\psi}(\Delta+2) - V_{\nu}^{\psi}(\Delta+1)] + \\
q_1p[V_{\nu}^{\psi}(\Delta+2) - V_{\nu}^{\psi}(\Delta+1)]\quad\Delta\geq1.
\end{multline*}
According to our assumption, $V_{\nu}^{\psi}(\Delta+1) - V_{\nu}^{\psi}(\Delta)$ is independent of $\Delta\geq1$. Hence, we can conclude that $V^{\psi}_{\nu+1}(\Delta+1) - V^{\psi}_{\nu+1}(\Delta)$ is also independent of $\Delta\geq1$. Then, by mathematical induction, we can conclude that $V^{\psi}(\Delta+1) - V^{\psi}(\Delta)$ is independent of $\Delta\geq1$. To calculate the constant $\sigma$, we have 
\begin{align*}
V^{\psi}(\Delta+1) &- V^{\psi}(\Delta) = \\
& \alpha + (1-q_1)(1-p)[V^{\psi}(\Delta+2) - V^{\psi}(\Delta+1)] + \\
& q_1p[V^{\psi}(\Delta+2) - V^{\psi}(\Delta+1)]\\
= & \alpha + (1-q_1)(1-p)\sigma + q_1p\sigma.
\end{align*}
Finally, we obtain
\[
\sigma = \frac{\alpha}{q_1+p-2q_1p}.
\]
\end{proof}
With Lemma~\ref{lem-zipfvfproperty} in mind, we can proceed with obtaining the policy induced by $V^{\psi}(s)$. Same as we did in the proof of Theorem~\ref{thm-optimalgeometric}, we define $V^{\psi,a}(s)$ as the expected cost resulting from taking action $a$ at state $s$. To determine the induced policy, we only need to determine the sign of $\delta V^{\psi}(s)\triangleq V^{\psi,0}(s) - V^{\psi,1}(s)$ for $s\in\mathcal{S}$. When $\delta V^{\psi}(s)<0$, the action suggested by the induced policy is $a=0$. Otherwise, $a=1$ is suggested. Without loss of generality, we let $V^{\psi}(0)=0$. Then, we distinguish between the following states.
\begin{itemize}
\item $s = (0,0,-1)$.
\begin{align*}
\delta V^{\psi}(0,0,-1) = & pV^{\psi}(1) - (1-q_1)pV^{\psi}(1) - q_1pV^{\psi}(1) \\
= & 0.
\end{align*}
\item $s = (\Delta,0,-1)$ where $\Delta\geq1$. When $t_{max}>2$, we have
\begin{align*}
\delta V^{\psi}(\Delta,0,-1) = & (1-p)V^{\psi}(\Delta+1) - \\
& (1-q_1)(1-p)V^{\psi}(\Delta+1)  - \\
& q_1pV^{\psi}(\Delta+1)\\
= & q_1(1-2p)V^{\psi}(\Delta+1)\\
\geq & 0.
\end{align*}
When $t_{max}=2$, we have
\begin{align*}
\delta V^{\psi}(\Delta,0,-1) = & (1-p)V^{\psi}(\Delta+1) - \\
& (1-q_1)(1-p)V^{\psi}(\Delta+1,1,1)  - \\
& q_1pV^{\psi}(\Delta+1).
\end{align*}
Since $V^{\psi}(\Delta+1,1,1)$ satisfies \eqref{eq-zipfoptimalproof4}, we have
\begin{equation}\label{eq-zipfoptimalproof5}
\begin{split}
V^{\psi}(\Delta+1,&1,1) - V^{\psi}(\Delta+1) =\\
& pV^{\psi}(\Delta+1)- (1-q_1)(1-p)V^{\psi}(\Delta+1)\\
& -q_1pV^{\psi}(\Delta+1)\\
= & (1-q_1)(2p-1)V^{\psi}(\Delta+1)\\
\leq & 0.
\end{split}
\end{equation}
Then, we know that $V^{\psi}(\Delta+1,1,1) \leq V^{\psi}(\Delta+1)$. Consequently,
\begin{align*}
\delta V^{\psi}(\Delta,0,-1) \geq & (1-p)V^{\psi}(\Delta+1) - \\
& (1-q_1)(1-p)V^{\psi}(\Delta+1) - \\
& q_1pV^{\psi}(\Delta+1)\\
= & q_1(1-2p)V^{\psi}(\Delta+1)\\
\geq & 0.
\end{align*}
\item $s = (0,t,0)$ where $1\leq t\leq t_{max}-1$.
\begin{align*}
\delta V^{\psi}(0,t,0) = & (1-q_{t+1})pV^{\psi}(1) + q_{t+1}pV^{\psi}(1) - \\
& (1-q_1)pV^{\psi}(1) - q_1pV^{\psi}(1)\\
= & 0.
\end{align*}
\item $s = (0,t,1)$ where $1\leq t\leq t_{max}-3$ and $t_{max}\geq4$.
\begin{align*}
\delta V^{\psi}(0,t,1) = & (1-q_{t+1})pV^{\psi}(1) + q_{t+1}(1-p)V^{\psi}(1) - \\
& (1-q_1)pV^{\psi}(1) - q_1pV^{\psi}(1)\\
= & q_{t+1}(1-2p)V^{\psi}(1)\\
\geq & 0.
\end{align*}
\item $s = (0,t_{max}-2,1)$ where $t_{max}\geq3$.
\begin{multline*}
\delta V^{\psi}(0,t_{max}-2,1) = \\
(1-q_{t_{max}-1})pV^{\psi}(1,t_{max}-1,1) + \\
q_{t_{max}-1}(1-p)V^{\psi}(1) - pV^{\psi}(1).
\end{multline*}
We notice that $V^{\psi}(1,t_{max}-1,1)$ satisfies \eqref{eq-zipfoptimalproof4}. Hence, replacing $V^{\psi}(1,t_{max}-1,1)$ with the corresponding expression yields
\begin{multline*}
\delta V^{\psi}(0,t_{max}-2,1) = \\
(1-q_{t_{max}-1})p[f(1)-\theta^{\psi} + pV^{\psi}(2)] + \\
(q_{t_{max}-1}-q_{t_{max}-1}p-p)V^{\psi}(1).
\end{multline*}
According to Lemma~\ref{lem-zipfvfproperty}, $V^{\psi}(2) = V^{\psi}(1) + \sigma$. Hence, we have
\[
\begin{split}
\delta V^{\psi}(0,&t_{max}-2,1)\\
 = & (1-q_{t_{max}-1})p[f(1)-\theta^{\psi} +p(V^{\psi}(1)+\sigma)] +\\
&  (q_{t_{max}-1}-q_{t_{max}-1}p-p)V^{\psi}(1)\\
= & [(p-1)(p-q_{t_{max}-1})-q_{t_{max}-1}p^2]V^{\psi}(1) + \\
& (1-q_{t_{max}-1})p(f(1)-\theta^{\psi} + p\sigma).
\end{split}
\]
We recall that the expected AoII $\theta^{\psi}$ is given by Theorem~\ref{prop-weakpreemptiveperformance}. Plugging the expressions for $\sigma$, $\theta^{\psi}$, and using the property that $V^{\psi}(1) = \frac{\theta^{\psi}-f(0)}{p}$ yield \eqref{eq-zipfoptimalproof2}.
\begin{figure*}[!t]
\normalsize
\begin{equation}\label{eq-zipfoptimalproof2}
\delta V^{\psi}(0,t_{max}-2,1) = \alpha\frac{(q_{t_{max}-1}-q_{t_{max}-1}p-p)+(1-q_{t_{max}-1})p(q_1+2p-2q_1p)^2}{(q_1+2p-2q_1p)(q_1+p-2q_1p)}.
\end{equation}
\hrulefill
\vspace*{4pt}
\end{figure*}
We notice that $\alpha>0$ and the denominator of \eqref{eq-zipfoptimalproof2} is positive. Hence, examining the sign of the numerator of \eqref{eq-zipfoptimalproof2} is sufficient to determine the sign of $\delta V^{\psi}(0,t_{max}-2,1)$. To this end, we define 
\begin{multline*}
\mathcal{Q}_1 \triangleq (q_{t_{max}-1}-q_{t_{max}-1}p-p)+\\
(1-q_{t_{max}-1})p(q_1+2p-2q_1p)^2.
\end{multline*}
Then, by Condition~\ref{con}, we know that $\mathcal{Q}_1\geq0$. Consequently, $\delta V^{\psi}(0,t_{max}-2,1)\geq0$.
\item $s = (0,t_{max}-1,1)$.
\begin{align*}
\delta V^{\psi}(0,t_{max}-1,1) = & (1-p)V^{\psi}(1) - \\
& (1-q_1)pV^{\psi}(1) - q_1pV^{\psi}(1)\\
= & (1-2p)V^{\psi}(1)\\
\geq & 0.
\end{align*}
\item $s = (\Delta,t,0)$ where $\Delta\geq1$ and $1\leq t
\leq t_{max}-1$. When $t_{max}>2$, we have
\begin{align*}
\delta V^{\psi}(\Delta,t,0) = & (1-q_{t+1})(1-p)V^{\psi}(\Delta+1) +\\
& q_{t+1}(1-p)V^{\psi}(\Delta+1) - \\
& (1-q_1)(1-p)V^{\psi}(\Delta+1) - \\
& q_1pV^{\psi}(\Delta+1)\\
= & (1-p)V^{\psi}(\Delta+1) - \\
& (1-q_1)(1-p)V^{\psi}(\Delta+1) - \\
& q_1pV^{\psi}(\Delta+1)\\
= & q_1(1-2p)V^{\psi}(\Delta+1)\\
\geq & 0.
\end{align*}
When $t_{max}=2$, we recall that $V^{\psi}(\Delta+1,1,1)\leq V^{\psi}(\Delta+1)$. Hence, we have
\begin{align*}
\delta V^{\psi}(\Delta,t,0) = & (1-q_{t+1})(1-p)V^{\psi}(\Delta+1) +\\
& q_{t+1}(1-p)V^{\psi}(\Delta+1) - \\
& (1-q_1)(1-p)V^{\psi}(\Delta+1,1,1) - \\
& q_1pV^{\psi}(\Delta+1)\\
\geq & (1-p)V^{\psi}(\Delta+1) - \\
& (1-q_1)(1-p)V^{\psi}(\Delta+1) -\\
& q_1pV^{\psi}(\Delta+1)\\
\geq & 0.
\end{align*}
\item $s = (\Delta,t,1)$ where $\Delta\geq1$, $1\leq t\leq t_{max}-3$, and $t_{max}\geq4$.
\begin{align*}
\delta V^{\psi}(\Delta,t,1)= & (1-q_{t+1})(1-p)V^{\psi}(\Delta+1) +\\
& q_{t+1}pV^{\psi}(\Delta+1) -\\
& (1-q_1)(1-p)V^{\psi}(\Delta+1) - \\
& q_1pV^{\psi}(\Delta+1)\\
= & (q_1-q_{t+1})(1-2p)V^{\psi}(\Delta+1).
\end{align*}
Since we assume that Condition~\ref{con} holds, we know that $\delta V^{\psi}(\Delta,t,1)\geq0$.
\item $s = (\Delta,t_{max}-2,1)$ where $\Delta\geq1$ and $t_{max}\geq3$.
\begin{multline*}
\delta V^{\psi}(\Delta,t_{max}-2,1) = \\
(1-q_{t_{max}-1})(1-p)V^{\psi}(\Delta+1,t_{max-1}-1,1) + \\
q_{t_{max}-1}pV^{\psi}(\Delta+1) - (1-q_1)(1-p)V^{\psi}(\Delta+1) - \\
q_1pV^{\psi}(\Delta+1).
\end{multline*}
We recall that $V^{\psi}(\Delta+1,t_{max}-1,1)$ satisfies \eqref{eq-zipfoptimalproof4}. Hence, replacing $V^{\psi}(\Delta+1,t_{max}-1,1)$ with corresponding expression yields
\begin{align*}
\delta V^{\psi}(\Delta,t_{max}&-2,1) = \\
& (1-q_{t_{max}-1})(1-p)(f(\Delta+1)-\theta^{\psi} +\\
& (1-p)V^{\psi}(0) + pV^{\psi}(\Delta+2)) +\\
& q_{t_{max}-1}pV^{\psi}(\Delta+1) - \\
& (1-q_1)(1-p)V^{\psi}(\Delta+1) - \\
& q_1pV^{\psi}(\Delta+1).
\end{align*}
We recall that $V^{\psi}(\Delta+1)-V^{\psi}(\Delta) = \sigma$ for $\Delta\geq1$. Then,
\begin{align*}
\delta V^{\psi}(\Delta&,t_{max}-2,1) =\\
& [(1-q_1)(2p-1)-p^2(1-q_{t_{max}-1})]V^{\psi}(1) + \\
& (1-q_{t_{max}-1})(1-p)(f(\Delta+1)-\theta^{\psi})+\\
& \{[(1-q_1)(2p-1)-p^2(1-q_{t_{max}-1})]\Delta+\\
& (1-q_{t_{max}-1})(1-p)p\}\sigma.
\end{align*}
Meanwhile, $V^{\psi}(1) = \frac{\theta^{\psi}-f(0)}{p}$. Hence, we have
\begin{align*}
& \frac{\delta V^{\psi}(\Delta,t_{max}-2,1)}{\alpha} = \\
& \quad \frac{(1-q_1)(2p-1)-p(1-q_{t_{max}-1})}{p}\theta^{\psi}+ \\
& \quad(1-q_{t_{max}-1})(1-p)(\Delta+1)+\\
& \quad \{[(1-q_1)(2p-1)-p^2(1-q_{t_{max}-1})]\Delta+\\
& \quad (1-q_{t_{max}-1})(1-p)p\}\sigma.
\end{align*}
We define the coefficient before $\Delta$ as $\mathcal{Q}_2$. Then, we have \eqref{eq-EquivalentEq4} holds.
\begin{figure*}[!t]
\normalsize
\begin{equation}\label{eq-EquivalentEq4}
\begin{split}
\mathcal{Q}_2= & [(1-q_1)(2p-1)-p^2(1-q_{t_{max}-1})]\sigma + (1-q_{t_{max}-1})(1-p)\\
= & \frac{(1-q_1)(2p-1)-p^2(1-q_{t_{max}-1})+ (1-q_{t_{max}-1})(1-p)(q_1+p-2q_1p)}{q_1+p-2q_1p}\\
= & \frac{(1-2p)\{q_1-1+(1-q_{t_{max}-1})[p+q_1(1-p)]\}}{q_1+p-2q_1p}.
\end{split}
\end{equation}
\hrulefill
\vspace*{4pt}
\end{figure*}
Then, under Condition~\ref{con}, we know that $\mathcal{Q}_2\geq0$. Hence, $\frac{\delta V^{\psi}(\Delta,t_{max}-2,1)}{\alpha}$ is non-decreasing in $\Delta$. Then, when $\Delta = 1$, we have \eqref{eq-EquivalentEq5} holds.
\begin{figure*}[!t]
\normalsize
\begin{equation}\label{eq-EquivalentEq5}
\begin{split}
\frac{\delta V^{\psi}(1,t_{max}-2,1)}{\alpha} = &  \frac{(1-q_1)(2p-1)-p(1-q_{t_{max}-1})}{(2p+q_1-2q_1p)(p+q_1-2q_1p)} +  \frac{(1-q_{t_{max}-1})(1-p)p}{q_1+p-2q_1p}+(1-q_{t_{max}-1})(1-p) + \mathcal{Q}_2\\
\triangleq & \mathcal{Q}_3.
\end{split}
\end{equation}
\hrulefill
\vspace*{4pt}
\end{figure*}
Again, under Condition~\ref{con}, we know that $\mathcal{Q}_3\geq0$. Combing with the fact that $\alpha>0$, we can conclude that $\delta V^{\psi}(1,t_{max}-2,1)\geq0$. Consequently, $\delta V^{\psi}(\Delta,t_{max}-2,1)\geq\delta V^{\psi}(1,t_{max}-2,1)\geq0$ for $\Delta\geq1$.
\item $s = (\Delta,t_{max}-1,1)$ where $\Delta\geq1$. When $t_{max}>2$, we have
\begin{align*}
\delta V^{\psi}(\Delta,t_{max}-1,1) =& pV^{\psi}(\Delta+1) - \\
& (1-q_1)(1-p)V^{\psi}(\Delta+1) - \\
& q_1pV^{\psi}(\Delta+1)\\
= & (1-q_1)(2p-1)V^{\psi}(\Delta+1)\\
\leq & 0.
\end{align*}
When $t_{max}=2$, we have
\begin{align*}
\delta V^{\psi}(\Delta,1,1) = & pV^{\psi}(\Delta+1) - \\
& (1-q_1)(1-p)V^{\psi}(\Delta+1,1,1) - \\
& q_1pV^{\psi}(\Delta+1).
\end{align*}
From \eqref{eq-zipfoptimalproof5}, we know that $V^{\psi}(\Delta+1,1,1) = (1-q_1)(2p-1)V^{\psi}(\Delta+1) + V^{\psi}(\Delta+1)$. Bring in the expression yields
\begin{multline*}
\delta V^{\psi}(\Delta,1,1) = \\
(1-q_1)[(2-2q_1)p^2+(3q_1-1)p-q_1]V^{\psi}(\Delta+1).
\end{multline*}
We recall that $0<p<\frac{1}{2}$ and $0\leq q_1\leq1$. Hence, $\delta V^{\psi}(\Delta,1,1)\leq0$.
\end{itemize}
Combining the above cases, we can conclude that the policy iteration algorithm converges to the threshold preemptive policy. Then, by Theorem~\ref{thm-policyimprovement}, the threshold preemptive policy is optimal.

\end{document}